\documentclass[a4paper, onecolumn,11pt]{article}
\usepackage[margin=0.8in]{geometry}
\pdfoutput=1
\usepackage[utf8]{inputenc}
\usepackage[english]{babel}
\usepackage[T1]{fontenc}
\usepackage{amsmath,amssymb,latexsym}
\usepackage{amsfonts}
\usepackage{calc} 
\usepackage[title]{appendix}
\usepackage[ruled]{algorithm}
\usepackage{algpseudocode}
\usepackage{listings}
\usepackage{enumerate}
\usepackage{subfig}      
\usepackage{float}
\usepackage[export]{adjustbox}
\lstnewenvironment{codice}[1][]
{\lstset{basicstyle=\ttfamily\small,
stringstyle=\ttfamily\small,
identifierstyle=\ttfamily\small,
keywordstyle=\color{blue}\bfseries\ttfamily\small, commentstyle=\color{gray}\ttfamily\small,
language=Haskell, basicstyle=\ttfamily\small, mathescape=true}}{}

\usepackage{graphicx}
\usepackage{caption}

\usepackage{multirow}
\usepackage[normalem]{ulem}

\usepackage{tikz} 
\usetikzlibrary{chains,scopes,fit}
\usetikzlibrary{arrows,automata,shapes,positioning,calc, matrix} 
\usepackage{circuitikz}

\usepackage{url}
\newtheorem{example}{Example}
\newtheorem{remark}{Remark}
\newtheorem{theorem}{Theorem}[section]
\newtheorem{proof}{Proof}[section]
\newtheorem{definition}{Definition}[section]

\usepackage[inference]{semantic}

\definecolor{mbc}{RGB}{255, 40, 40}

\usepackage{authblk}
\title{Quantum Markov Chain Semantics for Quip-E Programs}
\author[1,2]{Linda Anticoli}
\author[3]{Leonardo Taglialegne}
\affil[1]{University College London, WC1E 6BT London, United Kingdom.}
\affil[2]{Cambridge Quantum Computing Limited, CB2 1UB Cambridge, United Kingdom.}
\affil[3]{Dept. of Mathematics, Computer Science and Physics, University of Udine, Udine, Italy.}
\date{}                     
\date{}

\newtheorem{lemma}{Lemma}[section]

\begin{document}

\maketitle

\begin{abstract}
In this work we present a mapping from a fragment of the quantum programming language Quipper, called Quip-E,  to the
semantics of the QPMC model checker, aiming at the automatic verification of quantum programs. As a main outcome, we define a structural operational semantics for the Quip-E language corresponding to quantum Markov chains, and we use it as a basis for analysing quantum programs through the QPMC model checker. The properties of the semantics are proved and contextualised in the development of a tool translating from quantum programs to quantum Markov chains.
\end{abstract}

\section{Introduction}
\label{sec:introduction}

Quantum programming languages allow the specification of quantum programs in a human readable form, and their translation into machine executable code. In contrast with classical computation, where a huge variety of high-level programming languages are available, allowing the programmer to abstract from the physical details of both machine and the parading used, quantum programming languages were regarded as strictly related to the underlying physical hardware. Quantum programs are given in terms of quantum circuits which have a simple mathematical description but could be very difficult to realise in practice without the knowledge of the underlying physical model.  For this reason, a wide variety of different higher level programming languages have been provided (e.g., as in \cite{Omer2005, Zuliani2000QP, Gay2005, liquid1} among others). 
However, the use of higher level languages provides a good level of abstraction from the physical model, but does not guarantee the correctness of the quantum program.

Quantum model checking refers to a technique used to verify formal properties of quantum algorithms and protocols. Even in an idealised, noise free regime, quantum computation is based on the counter-intuitive laws of quantum physics, and it relies on a very fragile equilibrium in order to maintain a hopefully-fault-tolerant computation. Hence, the ability of validating and verifying quantum protocols, by assessing their correctness in order to avoid unexpected behaviours, is crucial. In the past, different authors explored theoretical proposals for quantum model checkers (e.g., as in \cite{BaltazarChadha2007, GayPapanikolaou2008} among others), using different structures as models and different temporal logics to express the properties. 

In this work we focus on quantum Markov chains (QMCs) as models: our aim is to provide a structural operational semantics in terms of QMCs for quantum programs written in a high level quantum programming language. In this way we want to formalise the work presented in \cite{entangle-valuetools}, in which a tool for the verification of quantum programs formal properties has been presented. 

In developing our framework, we used the functional language Quipper \cite{Selinger2014} and the PRISM-inspired model checking system QPMC \cite{QPMC}. 
Quipper is a functional quantum programming language based on Haskell that allows to build quantum programs and provides the possibility to simulate them. QPMC is a model checker for quantum protocols that uses the quantum temporal logic QCTL, an extension of PCTL \cite{kwiatowska1998}, in order to verify properties of quantum protocols.  Unfortunately, Quipper lacks a built-in formal verification tool. On the other hand, QPMC supports formal verification but it is based on a low-level specification language. 
In order to overcome the limitations of both frameworks, we built \texttt{Entang$\lambda$e}, a translator from quantum programs written in a Quipper sublanguage that we isolated (i.e., Quip-E ) into QPMC structures (i.e., quantum Markov chains). 

In this work we define a structural operational semantics for Quip-E programs in terms of QMCs, with the goal of formalising the translation and providing a way to adapt, or extend it outside QPMC, as a future goal.

The paper is organised as follows. In Section \ref{sec:preliminaries} we recall some basic quantum notations and the main mathematical formalisms used. Then we briefly introduce Quip-E and QPMC. In Section \ref{sec:TRCtranslation} we define the structural operational semantics allowing to translate Quip-E programs into QMCs. In Section \ref{sec:implementation} we introduce the translator \texttt{Entang$\lambda$e}, then we discuss some case studies used to verify our implementation by using simple examples of quantum algorithms. Section \ref{sec:conclusion} ends the paper summarising our contribution and outlining possible future lines of research.


\section{Preliminaries}\label{sec:preliminaries}

\subsection{Mathematical Quantum Models}\label{sec:quantum}

Quantum systems are represented by complex Hilbert spaces. A complex Hilbert space $\mathcal H$ is a complete vector space equipped with an inner product inducing a metric space. In particular, we will consider quantum systems described by finite dimensional Hilbert spaces of the form $\mathbb C^{2^k}$.
The elements of $\mathcal{H}$ (vectors) are denoted by either $\psi$ or $|\psi\rangle$ (i.e., ket notation). The notation $\langle \psi |$  (i.e., bra notation) denotes the transposed conjugate of $|\psi \rangle$.
The scalar product of two vectors $\varphi$ and $\psi$ in $\mathcal{H}$ is denoted by $\langle\varphi|\psi\rangle$, whereas $|\varphi\rangle \langle \psi |$ denotes the linear operator 
defined by $|\varphi\rangle$ and $\langle \psi|$.  We use $I$ to denote the identity matrix and $tr(\cdot)$ for the matrix trace. 

There are two possible formalisms based on Hilbert spaces for quantum systems: the \emph{state vector} formalism and the \emph{density matrix} one. We briefly summarise both of them, since Quipper is based on state vectors, while QPMC uses density matrices.

\subsubsection{State Vectors.}

\paragraph{State.} The \emph{state} of a quantum system is described by a \emph{normalised vector}  $|\psi\rangle \in \mathcal{H}$, i.e., $\lVert|\psi\rangle\rVert = \sqrt{\langle\psi|\psi\rangle}= 1$. 
The normalisation condition is  related to the probabilistic interpretation of quantum mechanics. 
\paragraph{Evolution.} The temporal \emph{evolution} of a quantum system is described by a \emph{unitary operator} (see, e.g., \cite{Nielsen-Chuang}). 

A linear operator $U$ is unitary if and only if  its conjugate transpose $U^{\dag}=(U^{T})^{*}$  coincides with its inverse $U^{-1}$.
Unitary operators preserve inner products and, as a consequence, norms of vectors. 
In absence of any measurement process, the state $|\psi_0\rangle$ at time $t_0$ evolves at time $t_1$ through the unitary operator $U$ to the state
$$|\psi_1\rangle = U \ |\psi_0\rangle$$

\paragraph{Measurement.} An \emph{observable} is a property of a physical system that can be measured, i.e., a physical quantity such as position, momentum, energy and spin. Observables are described by \emph{Hermitian 
operators} (see, e.g., \cite{Preskill}). 
A linear operator $A$ is Hermitian if $A = A^{\dag}$.  
Assuming non degeneracy, an 
Hermitian operator $A$ can be decomposed as 
$$A=\sum_{i=1}^n a_i | \varphi_i\rangle \langle \varphi_i|$$ where the $a_i$'s ($|\varphi_i\rangle$'s) are the eigenvalues (eigenvectors, respectively) of $A$.

Given a system in a state $|\psi\rangle$,
the outcome of a measurement of the observable $A$ is one of its eigenvalues $a_i$ and 
the state vector of the system after the measurement, provided the outcome $a_i$ has been obtained, is
\begin{equation*}\frac{(|\varphi_i\rangle\langle \varphi_i|) |\psi\rangle}{||(|\varphi_i\rangle\langle \varphi_i|) |\psi\rangle ||}\end{equation*}
with probability
\begin{equation*}p(a_i)  = || (|\varphi_i\rangle\langle \varphi_i|) |\psi\rangle ||^2 = \langle\psi|(|\varphi_i\rangle\langle \varphi_i|) |\psi\rangle\end{equation*}

\subsubsection{Density Matrices.}\label{sec:density}
\emph{Density matrices} take the role of state vectors. Quantum states described by state vectors  are idealised descriptions that cannot characterise statistical (incoherent) mixtures which often occur. These states are called \emph{mixed states}, and can be described by using density matrices.

 \paragraph{State. }The \emph{state} of a quantum system is described by an Hermitian, positive matrix $\rho$ with $tr(\rho)=1$.
Such matrices are called \emph{density} matrices.
A matrix $\rho$ is positive if for each vector $|\phi\rangle$ it holds that $\langle\phi|\rho|\phi\rangle \geq 0$.

Given a normalised vector $|\psi\rangle$ representing the state of a system through the state vector formalism, the corresponding density matrix is $|\psi\rangle \langle \psi|$.

\paragraph{Evolution and Measurement. }\emph{Evolutions} and \emph{measurements} of quantum systems are now described by \emph{superoperators} \cite{Nielsen-Chuang}. 
A superoperator is a (linear) function $\mathcal{E} : \rho_0 \to \rho_1$ which maps a
density matrix $\rho_0$ at time $t_0$ to a density matrix $\rho_1$ at time $t_1>t_0$
that satisfies the following properties:
$\mathcal{E}$ preserves hermiticity; 
$\mathcal{E}$ is trace preserving;
$\mathcal{E}$ is completely positive.

Let $\mathcal{B}_n (\mathcal{H})$ be the space of $n\times n$ density operators over $\mathcal{H}$. A linear map $\mathcal{E}: \mathcal{B}_n \mapsto \mathcal{B}_n $ is \emph{positive} if it maps positive operators into positive operators, while it is \emph{completely positive} if and only if $\mathcal{E} \otimes \mathbb{I}_m: \mathcal{B}_n \otimes \mathcal{B}_m \mapsto \mathcal{B}_n \otimes \mathcal{B}_m$ (where $\mathbb{I}_m$ is a $m$-dimensional identity operator) is positive for all $m \geq 0$.
Complete positivity is a requirement which allows a linear map to be positivity preserving even if the system under consideration (represented by an $n$-dimensional Hilbert space) has previously been correlated with another, unknown, system (represented by an $m$-dimensional Hilbert space). Positivity alone does not guarantee a positive evolution of the density matrix. 

Given a unitary operator $U$ the corresponding superoperator $\mathcal{E}_{U}$ can be defined as follows:
$$\mathcal{E}_{U}(\rho) = U\rho U^{\dag}$$

A quantum \emph{measurement} is described by a collection $\{M_i\}$ of linear operators, called measurement operators, satisfying the following condition:
\begin{equation*}\sum_i{M_i^\dag M_i} = I\end{equation*} 
The index \emph{i} refers to the measurement outcomes that may occur in the experiment. If $\rho$ is the state before the measurement and the outcome of the measurement is the \emph{i}-th one, then the state after the measurement is: 
\begin{equation*}\frac{M_i \rho M_i^\dag}{tr(M_i \rho M_i^\dag)}\end{equation*}
with probability
\begin{equation*}p(i) = tr(M_i \rho M_i^\dag)\end{equation*}
For example, given a measurement process of an observable $A=\sum_{i=1}^n a_i | \varphi_i\rangle \langle \varphi_i|$, the measurement operators are represented by $M_i=|\varphi_i\rangle \langle \varphi_i|$ where the index $i$ refers to the outcome $a_i$.

\subsection{Quip-E: a Quipper recursive fragment}\label{sec:quipper}

Quipper is an embedded functional programming language for quantum
computation \cite{Selinger2013} based on Knill's QRAM model \cite{knill-qram}. 

Quipper is endowed with a collection of data types,
combinators, and a library of functions within Haskell, together with
an idiom, i.e., a preferred style of writing embedded programs \cite{Selinger2013}.
It provides an extended quantum-classical circuit model which allows the use of quantum and 
classical wires (quantum and classical bits respectively) and quantum and classical gates (unitary and classical logic gates respectively) within
a circuit. 

Since Quipper is above all a circuit description language, 
it uses the state vector formalism and its main purpose is
to make circuit implementation easier providing high level operations for circuit manipulation.
A Quipper program is a function that inputs some quantum and classical data, performs
state changes on it, and then outputs the changed quantum/classical data. 
The quantum core of the computation
is encapsulated in a Haskell monad called \texttt{Circ}, which from an abstract point of view returns a quantum circuit.

A set of predefined gates (e.g., \texttt{hadamard}, \texttt{cnot}, \dots), together with the possibility of specifying \emph{ancilla} qubits and \emph{controls}, are provided. For a more in-depth description of Quipper we address the literature in \cite{quipper-intro}.

In our work  we are interested in classical operations and controls inside the
\texttt{Circ} monad allowing to define recursive functions too.  In  
Example \ref{ex2} we show part of the code of a Quipper recursive version of the quantum Fourier transform in the \texttt{Circ} monad (see \cite{quipper-intro}).
\begin{example}\label{ex2}
The function \texttt{qft'} computes the quantum Fourier transform of a list of qubits. If the list is empty, the circuit itself is empty. If the input is a list of one qubit, then the Hadamard gate is applied. The circuit for a list of $n + 1$ qubits applies the circuit for $n$ qubits to the last $n$ elements of the list, followed by a set of rotations over all $n + 1$ qubits.
\begin{lstlisting}[language=Haskell, breaklines=true, firstline=29, basicstyle=\scriptsize\ttfamily\linespread{0.5}, lastline=39]
qft' :: [Qubit] -> Circ [Qubit]
qft' [] = return []
qft' [x] = do
  hadamard x
  return [x]
qft' (x:xs) = do
  xs' <- qft' xs
  xs'' <- rotations x xs' (length xs')
  x' <- hadamard x
  return (x':xs'')
where ...
\end{lstlisting}
\end{example}

Quipper allows the use of Boolean operators and \texttt{if-then-else} statements with tests performed on Boolean parameters. 
The \texttt{dynamic\_lift} operator converts a bit to a Boolean parameter. Hence, the result of a measurement over a qubit can 
be stored in a bit, and then converted to a Boolean and used as guard in a test.
Moreover, boolean parameters can be used to initialize qubits through the \texttt{qinit} operator.
In Example \ref{ex3} we show how these can be combined inside a simple recursive Quipper circuit.

\begin{example}\label{ex3}
In the following example we show an instance of the quantum coin flipping: a qubit is initialized to $|0\rangle$, then the Hadamard gate is applied to it and it is measured. If the outcome is $0$, i.e. the value associated to the state $|0\rangle$, then the circuit is re--initialized, otherwise it terminates.
This is repeated until the result of the measurement is $1$. 
Hence, the circuit halts after an unpredictable number of iterations, i.e. ``coin flips'', returning the qubit $|1\rangle$. 
\begin{lstlisting}[language=Haskell, breaklines=true, firstline=320, basicstyle=\scriptsize\ttfamily\linespread{0.5}, lastline=330]
coinFlipCirc (q) = do
  q <- qinit[False]
  hadamard_at q
  m <- measure q
  bool <- dynamic_lift m
  if bool
    then
      return (q)
    else
      coinFlipCirc (q)
\end{lstlisting}
\end{example}
In the above example the circuit is tail recursive. In classical computation tail recursion corresponds to \texttt{while-loops} which, together with concatenation
of instructions, assignments, increments, and comparisons, give rise to a Turing-complete formalism.
In the case of quantum circuits tail recursion is probably the most natural form of recursion. A sequence of unitary gates is applied, the 
result is measured over some qubits and the result  is evaluated to decide whether repeat or stop the circuit.

The fragment of Quipper we are interested in is the \texttt{Circ} monad in which we allow tail recursion. 
In particular, we allow the use of the initialisation operator \texttt{reset}, of unitary operators, and of measurements, and we call this sublanguage Quip-E.
The results of measurements can be lifted to Boolean values and used inside a guard condition to decide whether to terminate the circuit or 
to restart it.

The \emph{body} \texttt{Body\_C} of a Quip-E \emph{tail-recursive program} $\texttt{trc\_C}$ is defined by the following grammar:

\begin{equation*}
\begin{split}
\texttt{Body\_C  ::= } &  \texttt{reset\_at q | U\_at [q$_{i_1}$,\dots,q$_{i_j}$] | m <- measure q | }\\
& \texttt{bool <- dynamic\_lift m | if (bool) Body\_C$_1$ else Body\_C$_2$ | }\\
& \texttt{Body\_C$_1$ Body\_C$_2$}
\end{split}
\end{equation*}
where \texttt{q, q$_{i_1}$,\dots,q$_{i_j}$} are qubits that occur as formal parameters of the program, \texttt{U\_at} is a unitary operator of dimension $j$, \texttt{m} is a bit variable name, \texttt{b} is a Boolean parameter. 

\begin{remark}
By nesting \texttt{if-then-else} constructors it is possible to mimic conditions that depend on every possible Boolean combination of sets of Boolean parameters. 
Hence, in the formal definition of the language we omit Boolean combinations without loosing expressive power, while our translator of Quip-E allows their use to ease the programming task. 
\end{remark}

A \emph{tail-recursive circuit} $\texttt{trc\_C}$ has the form  
\begin{lstlisting}[language=Haskell, breaklines=true, firstline=67, basicstyle=\scriptsize\ttfamily\linespread{0.5}, lastline=72]
trc_C :: (Qubit, Qubit,...) -> Circ RecAction
      trc_C (q1, q2, ...) = do
           -- beginning of body
           Body_C
           -- end of body
           exitOn bool
\end{lstlisting} 
where \texttt{q$_1$, q$_2$, \dots} are the qubits occurring in \texttt{Body\_C} and \texttt{bool} is a boolean parameter occurring in \texttt{Body\_C}. In this case we say that \texttt{trc\_C} is the Quip-E program defined by the body \texttt{Body\_C} and the exit condition \texttt{exitOn bool}.
Intuitively, the execution of \texttt{Body\_C} is repeated until \texttt{bool} becomes true. 

We impose that whenever a Boolean parameter \texttt{bool} is used as a guard of \texttt{if-then-else} and \texttt{exitOn} constructors, its value must have been previously defined in the body (e.g., through a \texttt{measure} instruction followed by a \texttt{dynamic\_lift}).

\begin{remark}
Non recursive programs can be defined as recursive ones using exit conditions that are always true. Hence, we omit them in the formal definition of Quip-E even if our translator allows their explicit use.
\end{remark}

\begin{example}\label{ex:toy}
The following is a small example of a Quip-E program. In the program two qubits are initialized to $|0\rangle$ and $|1\rangle$, respectively, then
Hadamard is applied to the second one, after the second qubit is measured and the result of the measurement is used both to decide which gate has to be applied to the first qubit and whether the program has to loop or terminate. 

\begin{lstlisting}[language=Haskell, breaklines=true, firstline=74, basicstyle=\scriptsize\ttfamily\linespread{0.5}, lastline=86]
exampleCirc :: (Qubit, Qubit) -> Circ RecAction
      exampleCirc (q1, q2) = do
          reset_at q1
          reset_at q2
          gate_X_at q2
          hadamard_at q1
          m <- measure q2
          bool <- dynamic_lift m
          if bool
              then gate_X_at q1
              else gate_Z_at q1
          m1 <- measure q1
          exitOn bool
\end{lstlisting} 

The same program can be written in Quipper native formalism as follows:
\begin{lstlisting}[language=Haskell, breaklines=true, firstline=88, basicstyle=\scriptsize\ttfamily\linespread{0.5}, lastline=98]
exampleCirc :: (Qubit, Qubit) -> Circ ()
      exampleCirc (q1, q2) = do
          [q1, q2] <- qinit [True, False]
          hadamard_at q1
          m <- measure q2
          bool <- dynamic_lift m
          if bool
              then gate_X_at q1
              else gate_Z_at q1
                   exampleCirc(q1,q2)
          m1 <- measure q1
\end{lstlisting} 

The \texttt{reset} function is a way to provide a unitary operator for the  \texttt{qinit} instruction. In particular, 
instruction
 \texttt{reset\_at q} in Quip-E is equivalent to the Quipper instruction \texttt{q <- qinit False}, which initializes the qubit to $|0\rangle$. 
 If \texttt{reset\_at q} is followed by the application of a \texttt{not} gate on \texttt{q} (e.g., \texttt{gate\_X\_at q}), then the sequence of two instructions of \emph{Quip\_E} is equivalent in Quipper to \texttt{q <- qinit True}, which initializes the qubit to $|1\rangle$. 
\end{example}

\subsection{QPMC: Quantum Program/Protocol Model Checker}

QPMC is a model checker for quantum programs and protocols based on the density matrix formalism
available in both web-based and off-line versions.\footnote{http://iscasmc.ios.ac.cn/too/qmc}
It takes in input programs written in
an extension of the PRISM guarded command language \cite{PRISM} that allows, in addition to the constants definable
in PRISM, the specification of the \texttt{vector}, \texttt{matrix}, and \texttt{superoperator} types.
QPMC supports the bra-ket notation and inner, outer and tensor product can be written using it.

The semantics of a QPMC program is given in terms of 
a \emph{superoperator weighted Markov chain}  -- a Markov chain in which the
state space is classical, while all quantum effects are encoded in the superoperators labelling the transitions 
(see, e.g., \cite{QPMC,MCQMC2013}). 
Differently from what we defined in Section \ref{sec:density}, QPMC superoperators are not necessarily trace-preserving, they are just completely positive linear operators.
A trace-non-increasing superoperator describes processes in which extra-information is obtained by \emph{measurement}. We briefly provide some definition useful to understand what follows (as given in \cite{MCQMC2013}).

\noindent Let $\mathcal{S(H)}$ be the set of superoperators over a Hilbert
space $\mathcal{H}$ and $\mathcal{S^I}(\mathcal{H})$ be the subset of
trace-nonincreasing superoperators.

\noindent Given  a density matrix $\rho$ representing the state of a system, 
$\mathcal{E} \in \mathcal{S^I(H)}$ implies that 
$tr(\mathcal{E}(\rho)) \in [0,1]$. Hence, it is natural to regard the set $\mathcal{S^I}(\mathcal{H})$ as the
quantum correspondent of the domain of traditional probabilities \cite{QPMC}.

\noindent Given two superoperators $\mathcal{E}, \ \mathcal{F} \in \mathcal{S(H)}$, $\mathcal{E} \mathbf{\lesssim} \mathcal{F}$ if for any quantum state $\rho$ it holds that $\ tr(\mathcal{E}(\rho)) \leq tr(\mathcal{F}(\rho))$. 

Informally, we can define a superoperator weighted \emph{Quantum Markov Chain} (herein QMC) as a discrete time Markov chain, where classical probabilities are replaced with quantum probabilities. Definition \ref{qmcdef} provides a more formal statement.
\begin{definition}[Quantum Markov Chain \cite{MCQMC2013,QPMC}]\label{qmcdef} A QMC over a Hilbert space $\mathcal{H}$ is a tuple \\ $\mathit{(S, Q, AP,
L)}$, where:
\begin{itemize}
\item $S$ is a finite set of classical states; 
\item $Q: S \times S \rightarrow \mathcal{S^I(H)}$
is the transition matrix where for each $s \in S$, the superoperator $\sum_{t\in S} {Q(s,t)}$ is trace-preserving
\item $AP $ is a finite set of atomic propositions
\item $L : S \to 2^{AP}$ is a labelling function
\end{itemize}
\end{definition}

The aim of QPMC is to provide a formal framework where to define and analyze properties of quantum protocols. 
The properties to be verified over QMC are expressed using the quantum computation tree logic
(QCTL), a temporal logic for reasoning about the evolution of quantum
systems introduced in \cite{MCQMC2013} that is a natural extension
of PCTL
.
\begin{definition}[Quantum Computation Tree Logic \cite{MCQMC2013,QPMC}]
A QCTL formula is a formula over the following grammar:
\begin{center}$SF ::= a \ | \  \lnot \ SF \ | \ SF \wedge SF \ | \ \mathbb{Q_{\sim \epsilon}}[PF] $\end{center}\begin{center}$PF ::= X \ SF \ | \ SF \ U^{\leq \mathit{k}} \ SF \ | \ SF \  U \ SF $\end{center}
where $a \in AP$, $\sim \ \in \{ \lesssim, \gtrsim, \eqsim \}$, 
$\mathcal{E} \in \mathcal{S^I}(\mathcal{H})$, $k \in \mathbb{N}$. 
$SF$ is a \emph{state }formula, while $PF$ is a \emph{path } formula. 
\end{definition}
The quantum operator formula $\mathbb{Q_{\sim \epsilon}}[PF]$ is
a more general case of the PCTL probabilistic operator $\mathbb{P_{\sim p}}[PF]$
and it expresses a constraint on the probability that the paths from
a certain state satisfy the formula $PF$. 
Besides the logical operators presented in QCTL, QPMC supports an
extended operator $Q=? [PF]$ to calculate (the matrix representation of) the superoperator
satisfying $PF$. Moreover,
QPMC provides the functions $qeval((Q=?)[PF], \rho)$ to compute the density operator obtained from applying
the resulting superoperator on a given density operator $\rho$,
and $qprob((Q=?)[PF], \rho) = tr(qeval((Q=?)[PF], \rho)))$ to calculate the probability of satisfying $PF$,
starting from the quantum state $\rho$ \cite{QPMC}.
\begin{example}\label{ex:toy-qpmc}
The following is a small example of a QPMC program.

\begin{lstlisting}[language=Haskell, breaklines=true, firstline=109, basicstyle=\scriptsize\ttfamily\linespread{0.5}, lastline=119]
module exampleCirc
  s: [0..3] init 0;
  b0: bool init false;

  [] (s = 0) -> <<kron(H, ID(2))>> : (s' = 1);
  [] (s = 1) -> <<kron(ID(2), M0)>> : (s' = 2) & (b0' = false) + <<kron(ID(2), M1)>> : (s' = 2) & (b0' = true);
  [] (s = 2) & b0 -> <<kron(PauliX, ID(2))>> : (s' = 3);
  [] (s = 2) & !b0 -> <<kron(PauliZ, ID(2))>> : (s' = 3);
  [] (s = 3) & !b0 -> true;
  [] (s = 3) & b0 -> true;
endmodule
\end{lstlisting} 
It is equivalent to the following Quipper circuit:
\begin{lstlisting}[language=Haskell, breaklines=true, firstline=123, basicstyle=\scriptsize\ttfamily\linespread{0.5}, lastline=130]
exampleCirc :: (Qubit, Qubit) -> Circ ()
exampleCirc (q1, q2) = do
  hadamard_at q1
  m <- measure q2
  b0 <- dynamic_lift m
  if b0
     then do gate_X_at q1
     else do gate_Z_at q1
\end{lstlisting} 
\end{example}

\section{Operational Semantics of Quip-E Programs}\label{sec:TRCtranslation}

In this section we provide an operational semantics for Quip-E programs; the semantics of such programs will be given in terms of QMCs. Intuitively,  a transition system through the operational semantics defines the operational rules  for all the programs in a given language. 
The nodes of such transition system represent the states during the computation and the transitions mimic the state changes. In the general case, the transition system associated with a program could have an infinite number of nodes. Even when it is finite, its size could depend on the input of the program, hence the transition system cannot be constructed on a generic input. 
In this Section we will see that the restrictions imposed on Quip-E ensure that we can associate a finite transition system to any Quip-E program. Such transition system turns out to be a QMC. It is important to note that in Quipper the use of lists of qubits together with 
recursion allows to represent an infinite family of circuits using a single program. The semantics we define in this section cannot be easily generalised to such Quipper programs.

Quip-E denotes a fragment of Quipper programs which generate only finite state, possibly circular, graphs of computations. Moreover, the dimension of such state spaces can be determined at compiling time. This is not the case if we consider generic Quipper programs having, for example, lists of qubits as formal parameters. In such cases, even if the state spaces are finite, their sizes depend on the length of the input qubits lists. 

Let us fix an a-priori finite set $\mathcal Q$ of qubits together with a finite set $\mathcal B$ of bits and Booleans. In this section we consider Quip-E programs whose variables and parameters are included in such sets. This assumption can be dropped, but this would make the description of the semantics more complex without increasing its expressibility.

Let \texttt{trc\_C} be a Quip-E program having body \texttt{Body\_C}. 
Let $\mathcal L$ be the set of functions from $\mathcal B$ to $\{0,1\}$. Intuitively, a function $L\in \mathcal L$ is an assignement of values for the bits and Booleans occurring in the program.
The rules in Table \ref{table} define by induction on the structural complexity of \texttt{Body\_C} its operational semantics in terms of QMCs. The states of such QMCs are pairs, whose first element is either the body of a program or the \emph{empty body}, denoted by \_\!\_\!\_. The second element of a pair is a function belonging to $\mathcal L$, which stores the current values of the bits and Booleans. All the operators that label the edges of the chain have dimension $2^{|\mathcal Q|}$.
Intuitively, if \texttt{Body\_C} is \texttt{reset\_at q$_k$}, then the qubit \texttt{q$_k$} is measured along the standard basis, applying the operators $\mathcal M_0^k$ and $\mathcal M_1^k$.
When $\mathcal M_0^k$ is applied the empty body is reached, when $\mathcal M_1^k$ is applied the body \texttt{X\_at q$_k$} is reached, and \texttt{X\_at} is the Pauli $X$ operator. In both cases there are no effects on the function $L$. In the case of \texttt{U\_at [q$_{i_1}$,\dots,q$_{i_j}$]} the superoperator 
$\mathcal U_{i_1,\dots,i_j}$
corresponding to \texttt{U} is applied and the empty body is reached, without affecting the function $L$. 
Such superoperator is computed by applying the identity operator to the qubits in $\mathcal Q\setminus \{ \texttt{q$_{i_1}$,\dots,q$_{i_j}$}\}$ and by swapping the qubits to preserve their order (see also Section \ref{sec:preliminaries}).
In the case of \texttt{m $\leftarrow$ measure q$_k$} the measure operators $\mathcal M_0^k$ and $\mathcal M_1^k$ are applied and the result of the measurement is stored by modifying $L(\texttt{m})$ accordingly. In particular, $L[L(\texttt{m})=i]$ denotes the function $L'$ which behaves as $L$ on $\mathcal B\setminus \{\texttt{m}\}$, while $L'(\texttt{m})$ has value $i$. In the case of \texttt{bool <- dynamic\_lift m} the identity superoperator $\mathcal I$ is applied, i.e., the qubits are unchanged, and the value stored in $L(\texttt{m})$ is copied in $L(\texttt{bool})$. 
In the case of an \texttt{if-then-else} instruction on the guard \texttt{bool} either the first or the second branch is chosen depending on the value of $L(\texttt{bool})$, without modifying the values of the qubits. In the case of a sequence \texttt{Body\_C$_1$ Body\_C$_2$} the first instruction of \texttt{Body\_C$_1$}, is executed applying the corresponding rule and the computation proceeds. Finally, the last rule is added only to ensure that also the empty body satisfies the second condition in the definition of QMC. 

\begin{table}[!h!t]
\begin{center}
\begin{tabular}{ccc}
\hline \\
\\
\inference{}{(\texttt{reset\_at q$_k$},L) \xrightarrow{\mathcal M_0^k} (\_\!\_\!\_,L)} & \qquad & \inference{}{(\texttt{reset\_at q},L) \xrightarrow{\mathcal M_1^k} ((\texttt{X\_at q$_k$},L)}\\
\\
\\
\multicolumn{3}{c}{\inference{}{(\texttt{U\_at [q$_{i_1}$,\dots,q$_{i_j}$]},L)\xrightarrow{\mathcal U_{i_1,\dots,i_j}} (\_\!\_\!\_,L)}}\\
\\
\\
\multicolumn{3}{c}{\inference{}{(\texttt{m $\leftarrow$ measure q$_k$},L) \xrightarrow{\mathcal M_i^k} (\_\!\_\!\_,L[L(\texttt{m})=i]\})}[\phantom{aaa}for $i\in \{0,1\}$]} 
\\
\\
\\
\multicolumn{3}{c}{\inference{}{(\texttt{bool <- dynamic\_lift m},L)\xrightarrow{\mathcal I}(\_\!\_\!\_,L[L(\texttt{bool})=L(\texttt{m})])}}\\
\\
\\
\multicolumn{3}{c}{\inference{L(\texttt{bool})=i}{(\texttt{if (bool) Body\_C$_1$ else Body\_C$_0$},L)\xrightarrow{\mathcal I}(\texttt{Body\_C$_i$},L)}[\phantom{aaa}for $i\in \{0,1\}$]}\\
\\
\\
\multicolumn{3}{c}{\inference{(\texttt{Body\_C$_1$},L) \xrightarrow{\mathcal S}(\texttt{Body\_C$_1$'},L')}{(\texttt{Body\_C$_1$ Body\_C$_2$},L) \xrightarrow{\mathcal S} 
(\texttt{Body\_C$_1$' Body\_C$_2$},L')}}\\
\\
\\
\multicolumn{3}{c}{\inference{(\texttt{Body\_C$_1$},L) \xrightarrow{\mathcal S}(\_\!\_\!\_,L')}{(\texttt{Body\_C$_1$ Body\_C$_2$},L) \xrightarrow{\mathcal S} 
(\texttt{Body\_C$_2$},L')}}
\\
\\
\\
\multicolumn{3}{c}{\inference{}{(\_\!\_\!\_,L)\xrightarrow{\mathcal I}(\_\!\_\!\_,L)}}
\\
\\
\\
\hline
\end{tabular}
\caption{Operational Semantics of Quip-E}\label{table}
\end{center}
\end{table} 

\begin{remark}\label{remark}
Notice that the values of the qubits are not stored in the state of the Markov chain. Their final values can be computed considering the composition of the operators which label the edges of the chain and by applying the resulting superoperator to their initial values (see \cite{QPMC}). In fact, all the operators that label the edges of the chain have dimension $2^{|\mathcal Q|}$ and it is fundamental that the order of the qubits is the same along all the chain. As a matter of fact, Quipper and Quip-E aim to provide a flexible programming framework and allow to specify at each step which are 
the qubits of interest and the order in which they enter a quantum gate. On the other hand, QMCs are a low level description language for quantum processes and as such they prefer minimality rather than flexibility. Hence, in a QMC all the gates are applied to all the qubits and these are always considered in the same order. This does not restrict the expressibility of QMC, since by exploiting swapping and identitiy operators it is always possible to extend a gate to all the qubits in the desired order.  
\end{remark}

Given its semantics, before we define the QMC associated to the body of a Quip-E program we need to provide further definitions.

\begin{definition}[Quantum Chains - $QC(s)$]
Let $s=(\emph{\texttt{Body\_C}},L)$ we define the structure $$QC(s)=(S(s), Q(s), AP(s), Lab(s))$$ as follows:
\begin{itemize}
\item $S(s)$ is the set of pairs reachable from $s$ by applying the rules of Table \ref{table};
\item $Q(s):S(s)\times S(s) \rightarrow \mathcal S^\mathcal I(\mathbb C^{2^{|\mathcal Q|}})$ is a transition operator defined by the rules of Table \ref{table}; 
\item $AP(s)=\mathcal B$; 
\item $Lab(s)((B',L'))=\{b\in \mathcal B \:|\: L'(b)=1\}$ are the labels allowing to keep track of the measurement results.
\end{itemize}
$QC(s)$ is said to be the \emph{quantum chain} of $s$.
\end{definition}

\begin{lemma}\label{pqc}
Given a state $s=(\emph{\texttt{Body\_C}},L)$ the quantum chain $QC(s)$ is a QMC.
\end{lemma} 

\begin{proof} 
In order to prove that the quantum chain $QC(s)$ is a QMC we have to verify that the sum of the superoperators labelling the edges outgoing from each state is a trace preserving superoperator. We proceed by cases as follows: 

\begin{enumerate}[(1)]
\item $s=$ (\texttt{U\_at [q$_{i1},\dots,$q$_{ij}$]}, $L ) \Rightarrow $\texttt{(s, L)}$ \xrightarrow{\mathcal{U}_{i_1,\dots,i_j}}$ \texttt{(\_\!\_\!\_,L)} is the only outgoing edge from $s$ and, since $\mathcal{U}_{i1,\dots,ij}$ is the superoperator associated to the unitary operator $U$, it is trivially trace preserving.

\item $s=$ (\texttt{m$\leftarrow$measure  [q$_{k}$]}, $L$) $ \Rightarrow Q(s)= \mathcal{M}_{i}^k, \ \ i\in\{0,1\}$. In this case there are two outgoing edges from $s$, labelled $\mathcal{M}_{0}^k$ and $\mathcal{M}_{1}^k$, i.e., the superoperators associated to the projection operators $M_0^k$ and $M_1^k$ respectively. The  property $\sum_i M_i^k = \mathbb{I}$ which follows from the definition of projection operator, can be lifted to the case of superoperators, hence $\sum_i \mathcal{M}_i^k=\mathcal{I}$ which verifies the requirement of the sum being trace preserving.

\item $s=$ (\texttt{reset\_at q$_{k}$}, $L$). This is as in case ($2$).

\item $s=$ (\texttt{bool$\leftarrow$dynamic\_lift  [m]}, $L$).  In this case there is only one outgoing edge with label $\mathcal{I}$.

\item $s=$ (\texttt{if (bool) Body\_C$_1$ else Body\_C$_2$}, $L$). As in case ($4$),  since $L$ satisfies either \texttt{L(bool)}$=true$ or \texttt{L(bool)}$=false$ but not both.

\item $s=$ (\texttt{Body\_C$_1$  Body\_C$_2$}, $L$). This follows by induction on \texttt{Body\_C}$_2$.

Since \texttt{Body\_C$_1$} and \texttt{Body\_C$_2$} are compositions of states as in ($1$)--($5$), the superoperator $\mathcal{S}$ is, with certainty, among the kinds already presented, hence it is trace preserving.

\item $s=$ (\texttt{\_\!\_\!\_}, $L$) $\Rightarrow Q(s)=\mathcal{I}$. As in case ($4$).
\end{enumerate}

\end{proof}

\begin{definition}[QMC associated to a body]
Let \emph{\texttt{Body\_C}} be a Quip-E body. The \emph{QMC associated to \texttt{Body\_C}}, denoted by $QC(\emph{\texttt{Body\_C}})$, is $$QC((\emph{\texttt{Body\_C}},O))$$ where $O$ is the function that assigns value $0$ to all the variables in $\mathcal B$.
\end{definition}

Notice that $QC(s)$ a part from the self-loops on the ``empty body states'', $QC(s)$ is acyclic. 
In order to define the semantics of Quip-E programs it is convenient to define an acyclic version of $QC(s)$ in which self-loops are removed.
\begin{definition}[Quasi QMC associated to a body]
Let \emph{\texttt{Body\_C}} be a Quip-E body. The \emph{Quasi QMC associated to \texttt{Body\_C}}, denoted by $QC^-(\emph{\texttt{Body\_C}})$, is the structure obtained by removing the self-loops in $QC(\emph{\texttt{Body\_C}})$. 
\end{definition}
The structure $QC^-(\emph{\texttt{Body\_C}})$ is not a QMC, since for the terminal states, i.e., the pairs whose first element is the empty body, the second condition of the definition of QMC is not satisfied. In the following definition we associate a QMC to a \emph{Quip\_E} program by introducing two rules that fix the violation. 

\begin{definition}[QMC associated to a program]
Let \emph{\texttt{trc\_C}} be a tail recursive \emph{Quip\_E} program defined by a body \emph{\texttt{Body\_C}} and an exit condition \emph{\texttt{exitOn bool}}.
The QMC associated to \emph{\texttt{trc\_C}}, denoted by $QC(\emph{\texttt{trc\_C}})$, is the QMC obtained from $QC^-(\emph{\texttt{Body\_C}})$ by adding the edges defined by the following rules:
\begin{center}
\begin{tabular}{ccc}
\inference{L(\emph{\texttt{bool}})=1}{(\_\!\_\!\_,L)\xrightarrow{\mathcal I}(\_\!\_\!\_,L)} & \qquad \qquad &

\inference{L(\emph{\texttt{bool}})=0}{(\_\!\_\!\_,L)\xrightarrow{\mathcal I}(\emph{\texttt{Body\_C}},O)}
\end{tabular}
\end{center}
\end{definition}
The following theorem states that our definition is correct, i.e., that the structure we associate to a program is a QMC.
\begin{theorem}
Let \emph{\texttt{trc\_C}} be a tail recursive \emph{Quip\_E} program. 
$QC(\emph{\texttt{trc\_C}})$ is a QMC.
\end{theorem}
\begin{proof}
In order to prove that $QC(s)$ is a QMC we have to verify that the sum of the superoperators labelling the edges outgoing from each state is a trace preserving superoperator. The first part follows from cases \emph{(1)--(6)} of the proof of Lemma \ref{pqc}, while case \emph{(7)} is replaced by two possibilities. Since either \texttt{L(bool)=1} or \texttt{L(bool)=0} are satisfied and they do not hold at the same time, there is always one outgoing edge with label $\mathcal{I}$, which is trace preserving.

\end{proof}

\subsection{Translation of \texttt{trc\_C} Programs}\label{sec:translation}
In the following we will show the translation of (the most significative) instructions of a \texttt{trc\_C} program into QMCs, according to their underlying operational semantics. For each of them we will provide a graphical representation of the resulting QMC as a directed graph, in which the nodes are the states of the chain, and the edges are labeled by the unitary or measurement superoperators, according to the order of the Quip-E instructions. The states are labelled according to the structural operational semantics defined.

The first four examples show the QMCs for single instructions, such as reset, unitary transformations, measurements and if--then--else respectively. 

For space reason, we will group ``quantum"-like instructions under the name \texttt{body}, and denote measurements, dynamic lift and conditional branches by using the line number associated; e.g., in Example \ref{ex8},  $\texttt{m <- measure q}=\texttt{1}$, $\texttt{b <- dynamic\_lift m}=\texttt{2}$, $\texttt{if bool}=\texttt{3}$, $\texttt{hadamard\_at q}=\texttt{4}$ and $\texttt{X\_at q}=\texttt{5}$.

\begin{example}
In this example we show a single--qubit \texttt{trc\_C} program in which a reset gate is applied. The Quip-E reset instruction together with its corresponding QMC, can be represented as follows:
\begin{lstlisting}[language=Haskell,numbers=left, numbersep=5pt, 
  numberstyle=\tiny,xleftmargin=2em,frame=single,framexleftmargin=1.5em]]
reset_at q
\end{lstlisting}
with $\mathcal{B}=\emptyset$.

\begin{center}
\scalebox{0.8}{
\begin{tikzpicture}[->,>=stealth',shorten >=1pt,auto,node distance=3cm,
  thick,main node/.style={ellipse,draw,minimum size=1.3cm,font=\sffamily\Large\bfseries}]
  \node[main node] (2)  {(\texttt{reset\_at q}, O)};
  \node[main node] (6) [right of=2, draw=white] {};
  \node[main node] (3) [right of=6] {(\_\!\_\!\_,O)};
  \node[main node] (4) [right of=3,draw=white] {};
  \node[main node] (5) [below =  1cm of 6,draw] {(\texttt{X\_at q}, O)};
  \path[every node/.style={font=\sffamily\small}]
    
    (2) edge [bend left=20] node [above] {$\mathcal{M}_0$} (3)
    (2) edge [bend right=20] node [below left] {$\mathcal{M}_1$} (5)
    (5) edge [bend right=20] node [below right] {$\mathcal{X}$} (3)
    (3) edge [loop right] node {$\mathcal{I}$} ();
\end{tikzpicture}}
\end{center}

\end{example}

\begin{example}
In this example we show a two--qubit program, in which an Hadamard gate is applied on the second qubit, then a measurement instruction is performed on the first one. The Quip-E instructions, together with the corresponding QMC, can be represented as follows:
\begin{lstlisting}[numbers=left, numbersep=5pt, 
  numberstyle=\tiny,xleftmargin=2em,frame=single,framexleftmargin=1.5em]
hadamard_at q2
m <- measure q1
\end{lstlisting}
with $\mathcal{B}=\{m\}$ and \texttt{O'=O[O(m)=1]}.

\begin{center}
\scalebox{0.6}{
\begin{tikzpicture}[->,>=stealth',shorten >=1pt,auto,node distance=3cm,
  thick,main node/.style={ellipse,minimum size=1.4cm,draw,font=\sffamily\Large\bfseries}]
  \node[main node] (0){(\texttt{body}, O)};
  \node[main node] (2) [right=2cm of 0] {(\texttt{m<- measure q1}, O)};
  \node[main node] (6) [right of=2, draw=white] {};
  \node[main node] (1) [above =1 cm of 6] {(\_\!\_\!\_,O)};
  \node[main node] (5) [below =1 cm of 6] {(\_\!\_\!\_,O')}; 

  \path[every node/.style={font=\sffamily\small}]
    (0) edge node {$\mathcal{I}\otimes\mathcal{H}$} (2)
    (2) edge [bend left=20] node [above left] {$\mathcal{M}_0 \otimes \mathcal{I}$} (1)
    (1) edge [loop right] node {$\mathcal{I}$} (1)
    (2) edge [bend right=20] node [below left] {$\mathcal{M}_1 \otimes \mathcal{I}$} (5)
    (5) edge [loop right] node {$\mathcal{I}$} (5);
\end{tikzpicture}}
\end{center}

\end{example}

\begin{example}
In this example we show a two--qubit program, in which two measurement are applied on the first and second qubit, respectively. The Quip-E instructions, together with the corresponding QMC, can be represented as follows:
\begin{lstlisting}[language=Haskell,numbers=left, numbersep=5pt, 
  numberstyle=\tiny,xleftmargin=2em,frame=single,framexleftmargin=1.5em]
m1 <- measure q1
m2 <- measure q2
\end{lstlisting}
with $\mathcal{B}=\{m1,m2\}$. 

\begin{center}
\scalebox{0.7}{
\begin{tikzpicture}[->,>=stealth',shorten >=1pt,auto,node distance=3cm,
  thick,main node/.style={ellipse,minimum size=1.4cm,draw,font=\sffamily\Large\bfseries}]

  \node[main node] (2) [right of=0] {(\texttt{1},O)};
  \node[main node] (6) [right of=2, draw=white] {};
  \node[main node] (1) [above =1 cm of 6] {(\texttt{2},O)};
  \node[main node] (3) [right of=6,draw=white] {};
  \node[main node] (4) [above of=3, right=1.5 cm, draw=white] {};
  \node[main node] (5) [below =1 cm of 6] {(\texttt{2},O')}; 
  \node[main node] (8) [below of=3, right=1.5 cm,draw=white] {};
  \node[main node] (ff) [above = -0.4 cm of 4] {(\_\!\_\!\_,O)};
  \node[main node] (ft) [below = 0.7 cm of 4] {(\_\!\_\!\_,O')};
  \node[main node] (tf) [above = 0.7 cm of 8] {(\_\!\_\!\_,O')};
  \node[main node] (tt) [below = -0.4 cm of 8] {(\_\!\_\!\_,O'')};

  \path[every node/.style={font=\sffamily\small}]
    (2) edge [bend left=20] node [above left] {$\mathcal{M}_0 \otimes \mathcal{I}$} (1)
    (2) edge [bend right=20] node [below left] {$\mathcal{M}_1\otimes \mathcal{I}$} (5)
    (1) edge node [above left] {$\mathcal{I}\otimes \mathcal{M}_0$} (ff)
    (ff) edge [loop right] node {$\mathcal{I}$} (ff)
    (1) edge node [below left] {$\mathcal{I}\otimes \mathcal{M}_1$} (ft)
    (ft) edge [loop right] node {$\mathcal{I}$} (ft)
    (5) edge node [above left] {$\mathcal{I} \otimes \mathcal{M}_0$} (tf)
    (tf) edge [loop right] node {$\mathcal{I}$} (tf)
    (5) edge node [below left] {$\mathcal{I}\otimes \mathcal{M}_1$} (tt)
    (tt) edge [loop right] node {$\mathcal{I}$} (tt);
\end{tikzpicture}}
\end{center}

\end{example}

\begin{example}\label{ex8}
In this example we show a single--qubit program, in which a measurement is performed and it is followed by a dynamic lifting, which transforms the resulting bit into a Boolean value, and by a conditional branch in which, according to the result an Hadamard or a Pauli X gate are applied. The Quip-E instructions, together with the corresponding QMC, can be represented as follows:
\begin{lstlisting}[language=Haskell,numbers=left, numbersep=5pt, 
  numberstyle=\tiny,xleftmargin=2em,frame=single,framexleftmargin=1.5em]
m <- measure q
b <- dynamic_lift m
if b
	then hadamard_at q
	else X_at q
\end{lstlisting}
with $\mathcal{B}=\{m,b\}$. 

\begin{center}
\scalebox{0.7}{
\begin{tikzpicture}[->,>=stealth',shorten >=1pt,auto,node distance=3.5cm,
  thick,main node/.style={ellipse,minimum size=1.4cm,draw,font=\sffamily\Large\bfseries}]

  \node[main node] (2) {(\texttt{1},O)};
  \node[main node] (6) [right of=2, draw=white] {};
  \node[main node] (1) [above =0.3 cm of 6] {(\texttt{2},O)};
  \node[main node] (3) [right of=6,draw=white] {};
  \node[main node] (5) [below =0.3 cm of 6] {(\texttt{2},O')};
  \node[main node] (7) [right of=1] {(\texttt{3},O)};
  \node[main node] (8) [right of=5] {(\texttt{3},O')};
  \node[main node] (9) [right of=7] {(\texttt{5},O)};
  \node[main node] (10) [right of=9] {(\_\!\_\!\_,O)};
  \node[main node] (11) [right of=8] {(\texttt{4},O')};
    \node[main node] (12) [right of=11] {(\_\!\_\!\_,O')};

  \path[every node/.style={font=\sffamily\small}]
    (2) edge [bend left=20] node [above left] {$\mathcal{M}_0$} (1)
    (2) edge [bend right=20] node [below left] {$\mathcal{M}_1$} (5)
    (1) edge node {$\mathcal{I}$} (7)
    (7) edge node {$\mathcal{I}$} (9)
    (9) edge node {$\mathcal{X}$} (10)
    (10) edge [loop right] node {$\mathcal{I}$} (10)
    (5) edge node {$\mathcal{I}$} (8)
    (8) edge node {$\mathcal{I}$} (11)
    (11) edge node {$\mathcal{H}$} (12)
    (12) edge [loop right] node {$\mathcal{I}$} (12);
\end{tikzpicture}}
\end{center}

\end{example}

In the following we present an example of \texttt{trc\_C} program in order to show the behaviour of the Quip-E tail--recursive instruction $exitOn$.

\begin{example}
In this example we show a two--qubit program, in which a measurement is performed and it is followed by a dynamic lifting, and by a conditional branch in which, according to the result an Hadamard or a Pauli X gate are applied. As last instruction, we have a recursive instruction which allows the program to terminate only when the Boolean value $\texttt{b}=true$. The Quip-E instructions, together with the corresponding QMC, can be represented as follows:
\begin{lstlisting}[language=Haskell,numbers=left, numbersep=5pt, 
  numberstyle=\tiny,xleftmargin=2em,frame=single,framexleftmargin=1.5em]
m1 <- measure q1
b1<- dynamic_lift m1
if b1
	then hadamard_at q2
	else X_at q2
exitOn b1
\end{lstlisting}
with $\mathcal{B}=\{m1,b1\}$. 

\begin{center}
\scalebox{0.7}{
\begin{tikzpicture}[->,>=stealth',shorten >=1pt,auto,node distance=3.5cm,
  thick,main node/.style={ellipse,minimum size=1.4cm,draw,font=\sffamily\Large\bfseries}]

  \node[main node] (2)  {(\texttt{1},O)};
  \node[main node] (6) [right of=2, draw=white] {};
  \node[main node] (1) [above =0.3 cm of 6] {(\texttt{2},O)};
  \node[main node] (3) [right of=6,draw=white] {};
  \node[main node] (5) [below =0.3 cm of 6] {(\texttt{2},O')};
  \node[main node] (7) [right of=1] {(\texttt{3},O)};
  \node[main node] (8) [right of=5] {(\texttt{3},O'')};
  \node[main node] (9) [right of=7] {(\texttt{5},O)};
  \node[main node] (10) [right of=9] {(\_\!\_\!\_,O)};
  \node[main node] (11) [right of=8] {(\texttt{4},O'')};
    \node[main node] (12) [right of=11] {(\_\!\_\!\_,O'')};

  \path[every node/.style={font=\sffamily\small}]
    (2) edge [bend left=20] node [above left] {$\mathcal{M}_0$} (1)
    (2) edge [bend right=20] node [below left] {$\mathcal{M}_1$} (5)
    (1) edge node {$\mathcal{I}$} (7)
    (7) edge node {$\mathcal{I}$} (9)
    (9) edge node {$\mathcal{X}$} (10)
    (10) edge [bend left=10, dashed] node {$\mathcal{I}$} (2)
    (5) edge node {$\mathcal{I}$} (8)
    (8) edge node {$\mathcal{I}$} (11)
    (11) edge node {$\mathcal{H}$} (12)
    (12) edge [loop right] node {$\mathcal{I}$} (12);
\end{tikzpicture}}
\end{center}
\end{example}

\section{Implementation}\label{sec:implementation}
The translator \texttt{Entang$\lambda$e} has been implemented in Haskell.  
In order to provide a more intuitive layout, \texttt{Entang$\lambda$e} has been provided with a web based graphical interface written in Elm,\footnote{\protect\url{http://elm-lang.org}.}. \texttt{Entang$\lambda$e}  is divided into three main blocks: \emph{Quipper}, \emph{Tree} and \emph{QPMC}, and a snapshot of its interface is shown in Fig. \ref{fig:entangle-full}.  

\begin{figure}[h!]
\begin{center}\includegraphics[width=\textwidth]{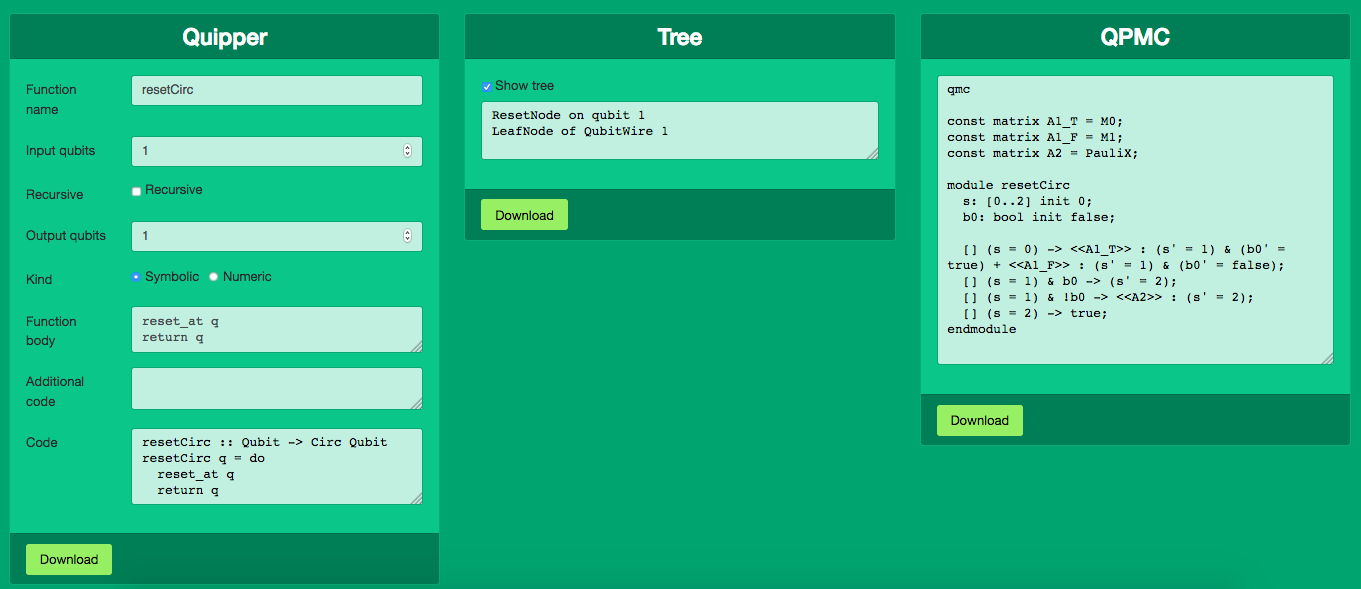}\end{center}
\protect\caption{Entang$\lambda$e graphic interface.}
\label{fig:entangle-full}
\end{figure}

A more in-depth description of  both \texttt{Entang$\lambda$e} and the underlying translation algorithm are provided in \cite{entangle-valuetools, RC2016}. 

\subsection{Experiments}\label{sec:experiments-short}

We tested \texttt{Entang$\lambda$e} with simple implementations of different programs, i.e., Deutsch--Jozsa, a Grover--based \emph{quantum switch} function, teleportation.
Grover's search and the BB84 quantum key distribution protocol were previously analysed in \cite{entangle-valuetools}. 

Where possible, for the aforementioned algorithms we provided both the recursive and non--recursive version. In the Appendix, for each algorithm we briefly recall its behaviour, we provide the Quip-E implementation and, for some of them, we provide their representation in terms of QMCs. The proper translation into QPMC code, due to space reasons, can be found in the Appendix.

\section{Conclusion}\label{sec:conclusion}
In this work we provided an operational semantics for Quip-E quantum programs, in order to formalise the underlying structure of \texttt{Entang$\lambda$e}, a framework for translating quantum programs into QPMC models, i.e., QMCs. The main idea is to create a tool for both writing quantum algorithms and protocols using a high-level programming language, and formally verifying them.
We put particular attention in the translation at a semantic level. While Quipper, and thus Quip-E, uses the state vector formalism and the quantum circuit model of computation, QPMC uses the density matrix formalism and QMCs, allowing to uniformly deal with both evolution and measurement operations.
We extended the tool in \cite{RC2016} in order to deal also with tail-recursive Quip-E programs, i.e., programs in which measurement results may lead to the termination or the re-execution of a particular circuit.

We aim at optimising our framework in order to validate complex algorithms and protocols, e.g., 
the ones using higher number of qubits, or involving a wider number of multi-qubit gates, which are actually difficult to translate due to the huge cost to generate the matrices. 
In the future, we also intend to  investigate  the specification of properties involving typical quantum effects, in particular automatic entanglement detection or multipartite entanglement representation, and to enhance the formal verification in the direction of symbolic model checking. 

Finally, even if the presented operational semantics is tailored on QPMC, in general QMCs can be used as underlying model for other model checkers, hence our work can be easily adapted and used outside QPMC.

\bibliographystyle{plain}

\section*{Acknowledgements}
We acknowledge useful discussion with Prof. Carla Piazza, and Prof. Paolo Zuliani, since both of them supported us with valuable insights.  A large part of this work was carried out when Dr. Linda Anticoli was affiliated with the University of Udine.
\section*{Conflict of interest}
The authors declare no conflict of interests.

\newpage

\begin{appendices}

\section{Experiments}

\subsection{Deutsch--Jozsa}

Let's consider a function $f$ from $n$ bits to $1$ bit, $f: \{0,1\}^n\longrightarrow \{0,1\}$.
Deutsch--Jozsa algorithm allows to distinguish between two different classes of functions, i.e., the constant and balanced ones. The function is constant if it evaluates the same on all inputs, i.e., the function is either $f(x) = 0$ or $f(x) = 1$ for every $x$, while it is balanced if the function is $0$ on one half of the possible inputs and $1$ on the other half. By using a quantum device, one single oracle query is needed to deterministically know whether the function is constant (the circuit output is a register $|0\rangle^{\otimes n}$) or not (the circuit output contains at least one state set to $|1\rangle$). Classically, the problem is solved by using $\frac{1}{2^n}+1$ queries, in the worst case since we have to apply the function on half the inputs instead of using a linear superposition of them.

 In the following we show an implementation of the algorithm using $3$--qubits oracles, plus an ancilla to implement it in a reversible way; the first oracle is constant and returns $0$ on all the inputs, while the second one is balanced.
 
 \paragraph{Implementation and Translation}
 In the following we present the Quip-E implementation for an instance of constant oracle, on the left, and a balanced one, on the right. 
 
 \begin{minipage}[t]{.45\textwidth}
\begin{lstlisting}[title=Constant,frame=tlrb,language=Haskell, breaklines=true, firstline=1, basicstyle=\scriptsize\ttfamily\linespread{0.5}, lastline=10]
dJozsaConst :: (Qubit, Qubit, Qubit, Qubit) -> Circ ()
dJozsaConst (q1, q2, q3, q4) = do
  map_reset_at (q1,q2,q3,q4)
  gate_X_at q4
  
  map_hadamard_at (q1,q2,q3, q4)
  
  map_hadamard_at (q1,q2,q3)
  measure (q1,q2,q3)
  return ()
\end{lstlisting}
\end{minipage}\hfill \ \ \ \
\begin{minipage}[t]{.45\textwidth}
\begin{lstlisting}[title=Balanced,frame=tlrb,language=Haskell, breaklines=true, firstline=1, basicstyle=\scriptsize\ttfamily\linespread{0.5}, lastline=12]
dJozsaBal :: (Qubit, Qubit, Qubit, Qubit) -> Circ ()
dJozsaBal (q1, q2, q3, q4) = do
  map_reset_at (q1,q2,q3,q4)
  gate_X_at q4
  
  map_hadamard_at (q1,q2,q3, q4)
  qnot_at q4 `controlled` [q1]
  qnot_at q4 `controlled` [q2]
  qnot_at q4 `controlled` [q3]
  map_hadamard_at (q1,q2,q3)
  measure (q1,q2,q3)
  return ()
\end{lstlisting}
\end{minipage}

 \paragraph{Test:}
 Some examples of QCTL formulae that we tested are presented in the following. In the case in which the output is a matrix, instead of displaying it, we put its representation in the Appendix, for space reasons. In order to show the probability associated, we provide the traces of the matrices.
 
 \begin{table}[h!]
  \begin{center}
    \label{tab:table1}
    \begin{tabular}{l|c|c}
      \textbf{QCTL Fomula} & \textbf{Output} & \textbf{Trace} \\ 
        &  & \\ 
      \hline
           \texttt{qeval(Q=? [F (s = 19 \& !b0 \& !b1 \& !b2)], r);} & (\ref{appendix}) & $1$ \\ 
           &  & \\ 
      \texttt{qeval(Q=? [F (s = 19 \& b0 \& !b1 \& !b2)], r);} & (\ref{appendix}) & $0$\\ 
      &  & \\ 
      \texttt{Q=1[F(s=19 \& !b0 \& !b1 \& !b2)];} & \texttt{true} &\\
      &  & \\ 
      \texttt{Q=1[F(s=19 \& b0 \& !b1 \& !b2)];} & \texttt{false} &\\
    \end{tabular}
        \caption{Deutsch--Jozsa Constant Verification.}
        \label{table: djc}
  \end{center}
\end{table}

The first formula computes the probability that, given an initial state \texttt{r}$=|0001\rangle\langle0001|$ we reach the final state $|0\rangle$. Such probability is equal to $1$, while the probability of reaching a final state in which at least a state $|1\rangle$ occurs is equal to $0$, as expected. In the following we consider the same queries in the case of the balanced oracle; the results change accordingly, since at least one state $|1\rangle$ should occur for the algorithm to success.

The last formula, in both Table \ref{table: djc} and \ref{table: djb}, investigates whether the probability of reaching the attended final state is equal to $1$, which is true since the two instances are deterministic.

 \begin{table}[h!]
  \begin{center}
    \label{tab:table1}
    \begin{tabular}{l|c|c}
      \textbf{QCTL Fomula} & \textbf{Output} & \textbf{Trace} \\ 
        &  & \\ 
      \hline
      \texttt{qeval(Q=? [F (s = 22 \& !b0 \& !b1 \& !b2)], r);} & (\ref{appendix}) & $0$ \\ 
      &  & \\ 
      \texttt{qeval(Q=? [F (s = 22 \& b0 \& b1 \& !b2)], r);} & (\ref{appendix}) & $1$\\ 
      &  & \\ 
      \texttt{Q>0.5[F(s=21 \& b0 \& b1 \& !b2)];} & \texttt{true} &\\
      &  & \\ 
      \texttt{Q<0.5[F(s=21 \& !b0 \& !b1 \& !b2)];} & \texttt{true} &\\
    \end{tabular}
        \caption{Deutsch--Jozsa Balanced Verification.}
        \label{table: djb}
  \end{center}
\end{table}

\subsection{Quantum Switch}

Classically, a Boolean switch function (or switch statement) checks for equality a discrete variable (or a Boolean expression) against a list of values, called cases. The variable to be switched  is checked for each case. Just as the classical version, a quantum switch returns, according to the value of the input qubits, the index of the correct gate (function) to be applied on them. In this way, we are sure that a given set of functions works properly on each possible combination of variables in input.

The idea behind the quantum switch is to use Grover's algorithm on a superposition of Boolean functions, represented by quantum \emph{oracles}, rather than on a superposition of basis states. A linear superposition of oracles is an operator which has the following matrix representation:
\begin{align}\hat{O}=\begin{pmatrix} U_0 & &  & \\
& U_1 & & \\
& & \ddots & \\
& & & U_n
\end{pmatrix} \equiv U_0 \oplus \dots \oplus U_n
\end{align}

While the aim of Grover's algorithm is searching for the index $i$ (represented by an $n$-bit string) of an element in an unstructured $N$-dimensional space, the aim of our Quantum Switch is to search for the index of the $i$-th Boolean function according to the index of the input qubits. 
In particular, the quantum switch relies on the diffusion operator of Grover's algorithm in order to amplify the probability that the right answer occurs. In order to do so, we have to extend the search space in order to provide in input both the variables to be switched and the linear superposition of oracles. Thus the search space will have size $N=2^{2n}+1$, with $2^n$ control qubits, i.e. the variables,  $2^n$ qubits on which Grover's diffusion operator is applied plus one ancillary qubit. Let us consider a non-trivial circuit with $N=16$, in Figure \ref{fig:qswitchns}
\begin{figure}[h]
\begin{center}\includegraphics[width=\textwidth]{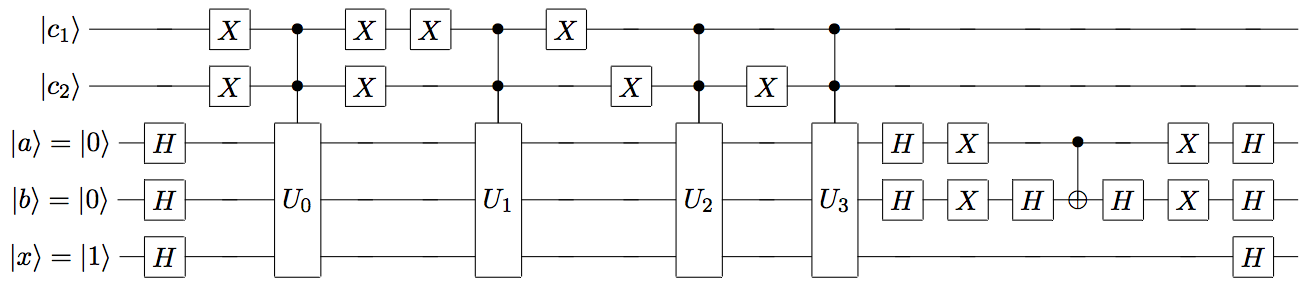}\end{center}
\protect\caption{Quantum Switch without superposition of variables.}
\label{fig:qswitchns}
\end{figure}
where, according to the values of $|c_1\rangle$ and $|c_2\rangle$, the circuit applies the correspondent oracle function and return in output the state $|c_1,c_2\rangle U_i |a,b\rangle$. In this particular case, the algorithm is deterministic and requires only one iteration of the diffusion operator. This is due to the fact that the Grover's diffusion operator is applied to a subspace of size $N_{sub}=4$.

A more interesting example can be found in Figure \ref{fig:qswitch}
\begin{figure}[h]
\begin{center}\includegraphics[width=\textwidth]{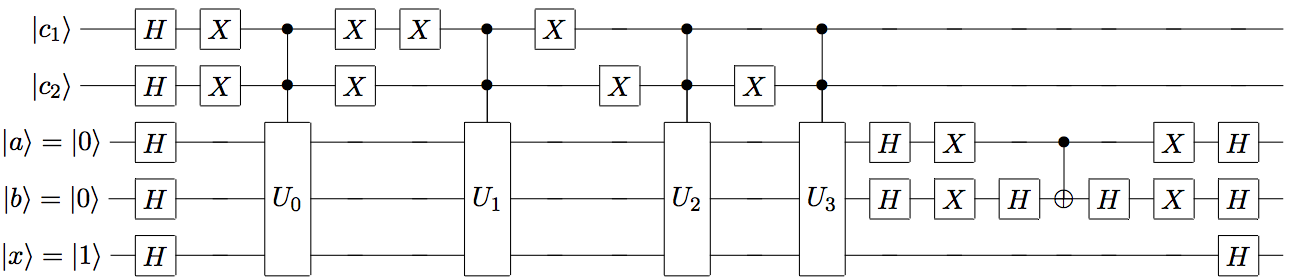}\end{center}
\protect\caption{Quantum Switch with Superposition of Variables.}
\label{fig:qswitch}
\end{figure}
\noindent where as input we provide a linear superposition of variables, thus we are considering all the possible inputs at the same time, and we want to check whether the output matches our expectations or not. The quantum switch exploits quantum parallelism to give to the superposition of quantum oracles  all the possible input strings at the same time. In the end we obtain in output a quantum state of the form $\frac{1}{2} \sum_{i,j} |c_i,c_j\rangle U_{i,j} |a,b\rangle$ which is an uniform distribution of oracles. In general, our quantum switch takes a linear superposition $\frac{1}{\sqrt{N}}|c_1,\dots,c_n, a_1,\dots, a_n\rangle$, applies a linear superposition of oracles  $U_0\oplus \dots \oplus U_{2^n}$ and returns a uniform distribution of states, in the form $\frac{1}{\sqrt{N}}\sum_{i_1,\dots,i_n}|c_{i_1},\dots,c_{i_n}\rangle U_{i_1,\dots,i_n}| a_{i_1},\dots, a_{i_n}\rangle$.

\subsubsection{Translation and Validation}
We translated an instance of the quantum switch algorithm, i.e., the deterministic one, with two variables and a search space of $N=32$, due to the ancillary qubit. Instances with more variables are still to be verified due to the high complexity of generating and performing verification of larger operators. 

\begin{minipage}[t]{.45\textwidth}
\begin{lstlisting}[title=Switch (1),frame=tlrb,language=Haskell, breaklines=true, firstline=1, basicstyle=\scriptsize\ttfamily\linespread{0.5}, lastline=11]
qswitchCirc :: (Qubit, Qubit, Qubit, Qubit, Qubit) -> Circ ()
qswitchCirc (q1, q2, q3, q4, q5) = do
  map_reset_at (q1,q2,q3,q4,q5)
  gate_X_at q5
  map_hadamard_at (q1,q2,q3,q4,q5)
  map_X_at (q1,q2,q3,q4)
  qnot_at q5 `controlled` [q3,q4, q1, q2]
  map_X_at (q1,q2,q3,q4)
  map_X_at (q1,q3)
  qnot_at q5 `controlled` [q3,q4, q1, q2]
  map_X_at (q1,q3)
\end{lstlisting}
\end{minipage}\hfill \ \ \ \
\begin{minipage}[t]{.40\textwidth}
\begin{lstlisting}[title=Switch (2),frame=tlrb,language=Haskell, breaklines=true, firstline=12, basicstyle=\scriptsize\ttfamily\linespread{0.5}, lastline=21]
  map_X_at (q2,q4)
  qnot_at q5 `controlled` [q3,q4, q1, q2]
  map_X_at (q2,q4)
  qnot_at q5 `controlled` [q3,q4, q1, q2]
  map_hadamard_at (q3,q4)
  map_X_at (q3,q4)
  hadamard_at q4
  qnot_at q4 `controlled` q3
  hadamard_at q4
  map_X_at (q3,q4)
\end{lstlisting}
\end{minipage}

\paragraph{Test:}

We tested some QCTL formulae:

\begin{table}[h!]\label{table: switch}
  \begin{center}
    \begin{tabular}{l|c}
      \textbf{QCTL Formula} & \textbf{Output}  \\ 
        &  \\ 
      \hline
      \texttt{Q>=0.25[F(s=39 \& !b0 \& !b1 \& !b2 \& !b3)];} & \texttt{true} \\
       \texttt{Q>=0.25[F(s=39 \& !b0 \& !b1 \& !b2 \& b3)];} & \texttt{false} \\ 
       \texttt{Q>=0.25[F(s=39 \& b0 \& b1 \& b2 \& b3)];} & \texttt{true} \\
       \texttt{Q>=0.25[F(s=39 \& b0 \& b1 \& b2 \& !b3)];} & \texttt{false} \\ 
    \end{tabular}
        \caption{Quantum Switch Verification.}
  \end{center}
\end{table}
The formulae verify that, in the future, a desired state (we restricted the example to two solutions for space reasons) is reached with probability bounded by $0.25$, while the probability to reach other, undesired states, is less than $0.25$, validating our expectations.

A larger version of the circuit, with a search space of $N=512$ can be found at \url{https://github.com/miniBill/entangle/tree/master/res/Entangle\_Tests}.

\subsection{Teleportation}
The aim of quantum teleportation is to move a qubit from one location to another, without physically transporting or copying it, with the aid of a classical channel and a shared quantum entanglement pair between the sender and the receiver. The teleportation protocol between two parties, namely Alice and Bob, can be summarised as follows:
first, an entangled pair is generated between Alice and Bob, then  Alice performs a Bell measurement (a specific sequence of unitary operators followed by a measurement) of her part of the entangled pair qubit and the qubit to be teleported. The measurement yields one of four measurement outcomes, which are then encoded using two classical bits. By using the classical communication channel, Alice sends the two bits to Bob. As a last step, according to the two received bits, Bob applies a pre--determined sequence of unitary gates on its part of the entangled pair, obtaining always the qubit that was chosen for teleportation. 
In the following, since the protocol always succeeds in absence of noise, we provide both the non--recursive and the recursive versions.

\begin{lstlisting}[title=Teleportation,frame=tlrb,language=Haskell, mathescape=false,breaklines=true, firstline=88, basicstyle=\scriptsize\ttfamily\linespread{0.5}, lastline=107]
teleport :: (Qubit, Qubit, Qubit) -> Circ ()
teleport (q1, q2, q3) :: do
    reset_at q2
    reset_at q3
    hadamard_at q2
    qnot_at q3 `controlled` q2
    qnot_at q2 `controlled` q1
    hadamard_at q1
    c1 <- measure q1
    c2 <- measure q2
    if c1==1 && c2==1
        then do
          gate_Y_at q3
        else if c1==1 && c2==0
          then do
            gate_Z_at q3
          else if c1==0 && c2==1
            then do
              gate_X_at q3
            else return ()
\end{lstlisting}

The first Hadamard gate, followed by a controlled--not is used to create a maximally entangled state between the second and the third qubit.  
A representation of the QMC for the teleportation version can be seen in Figure \ref{fig:tp}.

\begin{center}
\begin{figure}[h]
\resizebox{\textwidth}{!}{
\begin{tikzpicture}[->,>=stealth',shorten >=1pt,auto,node distance=1.3cm,
  thick,main node/.style={ellipse,draw,minimum size=0.5cm,font=\sffamily\footnotesize\bfseries}]
\tikzset{blue dotted/.style={draw=blue!50!white, line width=1pt,
                               dash pattern=on 1pt off 4pt on 6pt off 4pt,
                                inner sep=2mm, rectangle, rounded corners}};
\tikzset{gray dotted/.style={draw=gray!50!white, line width=1pt,
                               dash pattern=on 1pt off 4pt on 6pt off 4pt,
                                inner sep=2mm, rectangle, rounded corners}};
\tikzset{gray dashed/.style={draw=gray!50!white, line width=1pt,
                               dash pattern=on 1pt off 1pt on 1pt off 1pt,
                                inner sep=2mm, rectangle, rounded corners}};
\tikzset{multi/.style={ellipse,draw,minimum size=0.2cm,font=\sffamily\tiny\bfseries}} 

  \node[main node] (0) [draw=white]{};
  \node[main node] (s0) [right of=0] {$s_{0}$};
  \node[main node] (s1t) [above right =0.2cm and 0.5cm of s0,draw] {$s_{1}t$};
  \node[main node] (s2) [below right =0.2cm and 0.5cm of s1t] {$s_2$};
  \node[main node] (s1f) [below right =0.2cm and 0.5cm of s0,draw] {$s_{1}f$};
  \node[main node] (s3t) [above right =0.2cm and 0.5cm of s2,draw] {$s_{3}t$};
  \node[main node] (s4) [below right =0.2cm and 0.5cm of s3t] {$s_4$};
  \node[main node] (s3f) [below right =0.2cm and 0.5cm of s2,draw] {$s_{3}f$};
  \node[main node] (s5) [right=0.5 of s4] {$s_5$};
  \node[main node] (s6) [right=0.5 of s5] {$s_6$};
  \node[main node] (s7) [right=0.5 of s6] {$s_7$};
  \node[main node] (s8) [right=0.5 of s7] {$s_8$};
  \node[main node] (s9t) [above right = of s8] {$s_{9}t$};
  \node[main node] (s9f) [below right = of s8] {$s_{9}f$};
  \node[main node] (s10ff) [below right =0.3cm and 0.5cm of s9f] {$s_{10}ff$};
  \node[main node] (s10ft) [above right =0.3cm and 0.5cm of s9f] {$s_{10}ft$}; 
  \node[main node] (s10tf) [below right =0.3cm and 0.5cm of s9t] {$s_{10}tf$};
  \node[main node] (s10tt) [above right =0.3cm and 0.5cm of s9t] {$s_{10}tt$}; 
  \node[main node] (s11ff) [right =0.5 cm of s10ff] {$s_{11}ff$};
  \node[main node] (s11ft) [right =0.5 cm of s10ft] {$s_{11}ft$}; 
  \node[main node] (s11tf) [right =0.5 cm of s10tf] {$s_{11}tf$};
  \node[main node] (s11tt) [right =0.5 cm of s10tt] {$s_{11}tt$}; 
   \node[main node] (s12fff) [below right =0.5cm and 0.6cm of s11ff] {$s_{12}fff$};
  \node[main node] (s12fft) [right =0.5 cm of s11ff] {$s_{12}fft$}; 
  \node[main node] (s12ftf) [below right =0.3cm and 0.5cm of s11ft] {$s_{12}ftf$};
  \node[main node] (s12ftt) [right =0.5 cm of s11ft] {$s_{12}ftt$}; 
   \node[main node] (s12tff) [right =0.5 cm of s11tf] {$s_{12}tff$};
  \node[main node] (s12tft) [above right =0.2cm and 0.5cm of s11tf] {$s_{12}tft$}; 
  \node[main node] (s12ttf) [right =0.5 cm of s11tt] {$s_{12}ttf$};
  \node[main node] (s12ttt) [above right =0.5cm and 0.6cm of s11tt] {$s_{12}ttt$}; 
  
  
  \path[every node/.style={font=\sffamily\tiny}]
    (s1t) edge [bend left=20] node [above] {$\mathbb{I}$} (s2)
    (s0) edge [bend left=20] node [above] {$M_0 $} (s1t)
    (s0) edge [bend right=20] node [below] {$M_1$} (s1f)
    (s1f) edge [bend right=20] node [below] {$X$} (s2)
    (s2)edge [bend left=20] node [above] {$M_0$} (s3t)
    (s2) edge [bend right=20] node [below] {$M_1$} (s3f)
    (s3t) edge [bend left=20] node [above] {$\mathbb{I}$} (s4) 
    (s3f) edge [bend right=20] node [below] {$X$} (s4)
    (s4) edge node [above] {$H$} (s5)
    (s5) edge node [above] {$CN$} (s6) 
    (s6) edge node [above] {$CN$} (s7)
    (s7) edge node [above] {$H$} (s8)
   (s8) edge [bend right=20] node [above ] {$M_0$} (s9f)
    (s8) edge [bend left=20] node [below ] {$M_1$} (s9t)
    (s9f) edge [bend right=20] node [above ] {$M_0$} (s10ff)
    (s9f) edge [bend left=20] node [below ] {$M_1$} (s10ft)
    (s9t) edge [bend right=20] node [above ] {$M_0$} (s10tf)
    (s9t) edge [bend left=20] node [below ] {$M_1$} (s10tt)
    (s10ff) edge node [above ] {$I$} (s11ff)
    (s10ft) edge node [above ] {$X$} (s11ft)
    (s10tf) edge node [below ] {$Z$} (s11tf)
    (s10tt) edge node [below ] {$Y$} (s11tt)
    (s11ff)     edge [bend right=20] node [below ] {$M_0$} (s12fff)
    (s11ff)     edge [bend left=20] node [below ] {$M_1$} (s12fft)
    (s11ft)     edge [bend right=20] node [below ] {$M_0$} (s12ftf)
    (s11ft)     edge [bend left=20] node [below ] {$M_1$} (s12ftt)
    (s11tf)     edge [bend right=20] node [below ] {$M_0$} (s12tff)
    (s11tf)     edge [bend left=20] node [above ] {$M_1$} (s12tft)
       (s11tt)     edge [bend right=20] node [below ] {$M_0$} (s12ttf)
    (s11tt)     edge [bend left=20] node [above ] {$M_1$} (s12ttt)
        (s12ttt)      edge [loop right]  node [right]  {$I$}     ()
        (s12ttf)      edge [loop right]  node [right]  {$I$}     ()
        (s12tft)      edge [loop right]  node [right]  {$I$}     ()
        (s12tff)      edge [loop right]  node [right]  {$I$}     ()
        (s12ftt)      edge [loop right]  node [right]  {$I$}     ()
        (s12fft)      edge [loop right]  node [right]  {$I$}     ()
        (s12ftf)      edge [loop right]  node [right]  {$I$}     ()
        (s12fff)      edge [loop right]  node [right]  {$I$}     ();

\end{tikzpicture}}
\protect\caption{QMC for the Teleportation Protocol.}
\label{fig:tp}
\end{figure}
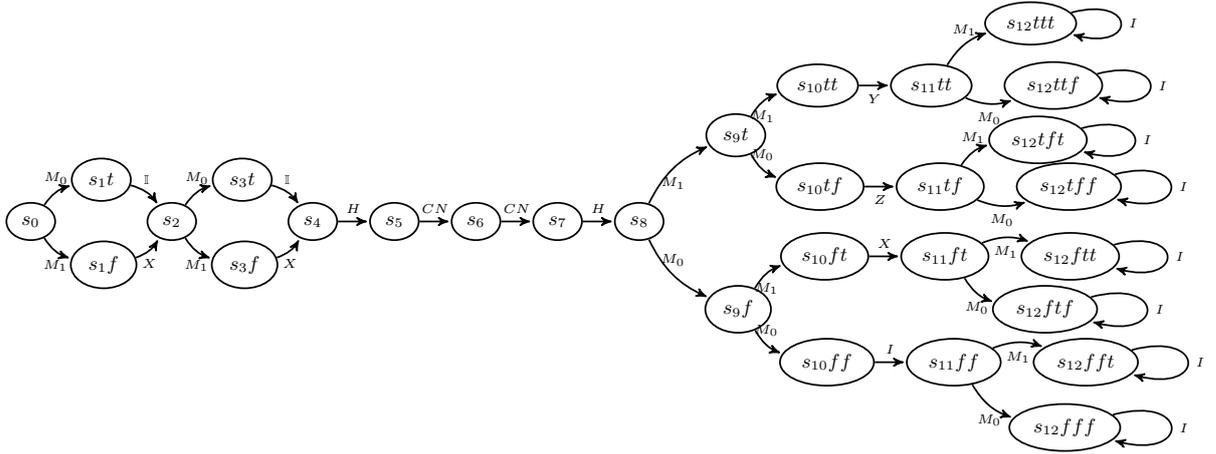
\end{center}

In the following we provided a tail--recursive version of the same protocol.
\vskip 1cm
\begin{lstlisting}[title=Teleportation Recursive,frame=tlrb,language=Haskell, mathescape=false, breaklines=true, firstline=60, basicstyle=\scriptsize\ttfamily\linespread{0.5}, lastline=86]
teleportRec :: (Qubit, Qubit, Qubit) -> Circ RecAction
teleportRec (q1, q2, q3) = do
    reset_at q2
    reset_at q3
    hadamard_at q2
    qnot_at q3 `controlled` q2
    qnot_at q2 `controlled` q1
    hadamard_at q1
    c1 <- measure q1
    c2 <- measure q2
    b1 <- dynamic_lift c1
    b2 <- dynamic_lift c2
    if b1 && b2
        then do
          gate_Y_at q3
        else if b1 && (not b2)
          then do
            gate_Z_at q3
          else if (not b1) && b2
            then do
              gate_X_at q3
            else do
              gate_X_at q3
              gate_X_at q3
    c3 <- measure q3
    b3 <- dynamic_lift c3
    exitOn $ b1==b3
\end{lstlisting}

\begin{table}[h!]\label{table: tel}
  \begin{center}
    \begin{tabular}{c|c}
      \textbf{QCTL Formula} & \textbf{Output}  \\ 
        &  \\ 
      \hline
      \texttt{Q>=0.25[F(s=11 \& !b0 \& !b1)];} & \texttt{true} \\
      \texttt{Q=0[F(s=12 \& !b0 \& !b1 \& !b2)];} & \texttt{false} \\ 
    \end{tabular}
        \caption{Teleportation Protocol  Verification.}
  \end{center}
\end{table}
The  formulae bounds the probability to reach a desired state (after the first conditional branch) to a value greater or equal  to $0.25$, and ascertain that the probability to reach a final, desired state (after the second conditional branch), never goes to $0$.

\newpage
\section{Results}\label{appendix}
In the following we show the examples of translated quantum algorithms, together with some results that have been omitted for space reason from the previous sections.

Exploiting our implementation of \texttt{Entang$\lambda$e} the following code has been automatically generated. 

\subsection{Deutsch--Jozsa Constant Oracle QMC}

\begin{lstlisting}[language=Haskell, breaklines=true, basicstyle=\tiny\ttfamily\linespread{0.5}, frame=single, framexleftmargin=1.5em]
qmc

const matrix A1_T = kron(M0, ID(8));
const matrix A1_F = kron(M1, ID(8));
const matrix A2 = kron(PauliX, ID(8));
const matrix A3_T = kron(kron(ID(2), M0), ID(4));
const matrix A3_F = kron(kron(ID(2), M1), ID(4));
const matrix A4 = kron(kron(ID(2), PauliX), ID(4));
const matrix A5_T = kron(kron(ID(4), M0), ID(2));
const matrix A5_F = kron(kron(ID(4), M1), ID(2));
const matrix A6 = kron(kron(ID(4), PauliX), ID(2));
const matrix A7_T = kron(ID(8), M0);
const matrix A7_F = kron(ID(8), M1);
const matrix A8 = kron(ID(8), PauliX);
const matrix A9 = kron(ID(8), PauliX);
const matrix A10 = kron(Hadamard, ID(8));
const matrix A11 = kron(kron(ID(2), Hadamard), ID(4));
const matrix A12 = kron(kron(ID(4), Hadamard), ID(2));
const matrix A13 = kron(ID(8), Hadamard);
const matrix A14 = kron(Hadamard, ID(8));
const matrix A15 = kron(kron(ID(2), Hadamard), ID(4));
const matrix A16 = kron(kron(ID(4), Hadamard), ID(2));
const matrix A17_F = kron(M0, ID(8));
const matrix A17_T = kron(M1, ID(8));
const matrix A18_FF = kron(kron(ID(2), M0), ID(4));
const matrix A18_TF = kron(kron(ID(2), M1), ID(4));
const matrix A19_FFF = kron(kron(ID(4), M0), ID(2));
const matrix A19_TFF = kron(kron(ID(4), M1), ID(2));
const matrix A19_FTF = kron(kron(ID(4), M0), ID(2));
const matrix A19_TTF = kron(kron(ID(4), M1), ID(2));
const matrix A18_FT = kron(kron(ID(2), M0), ID(4));
const matrix A18_TT = kron(kron(ID(2), M1), ID(4));
const matrix A19_FFT = kron(kron(ID(4), M0), ID(2));
const matrix A19_TFT = kron(kron(ID(4), M1), ID(2));
const matrix A19_FTT = kron(kron(ID(4), M0), ID(2));
const matrix A19_TTT = kron(kron(ID(4), M1), ID(2));

module dJozsaConst
  s: [0..19] init 0;
  b0: bool init false;
  b1: bool init false;
  b2: bool init false;

  [] (s = 0) -> <<A1_T>> : (s' = 1) & (b0' = true) + <<A1_F>> : (s' = 1) & (b0' = false);
  [] (s = 1) & b0 -> (s' = 2);
  [] (s = 1) & !b0 -> <<A2>> : (s' = 2);
  [] (s = 2) -> <<A3_T>> : (s' = 3) & (b0' = true) + <<A3_F>> : (s' = 3) & (b0' = false);
  [] (s = 3) & b0 -> (s' = 4);
  [] (s = 3) & !b0 -> <<A4>> : (s' = 4);
  [] (s = 4) -> <<A5_T>> : (s' = 5) & (b0' = true) + <<A5_F>> : (s' = 5) & (b0' = false);
  [] (s = 5) & b0 -> (s' = 6);
  [] (s = 5) & !b0 -> <<A6>> : (s' = 6);
  [] (s = 6) -> <<A7_T>> : (s' = 7) & (b0' = true) + <<A7_F>> : (s' = 7) & (b0' = false);
  [] (s = 7) & b0 -> (s' = 8);
  [] (s = 7) & !b0 -> <<A8>> : (s' = 8);
  [] (s = 8) -> <<A9>> : (s' = 9);
  [] (s = 9) -> <<A10>> : (s' = 10);
  [] (s = 10) -> <<A11>> : (s' = 11);
  [] (s = 11) -> <<A12>> : (s' = 12);
  [] (s = 12) -> <<A13>> : (s' = 13);
  [] (s = 13) -> <<A14>> : (s' = 14);
  [] (s = 14) -> <<A15>> : (s' = 15);
  [] (s = 15) -> <<A16>> : (s' = 16);
  [] (s = 16) -> <<A17_F>> : (s' = 17) & (b0' = false) + <<A17_T>> : (s' = 17) & (b0' = true);
  [] (s = 17) & b0 -> <<A18_FT>> : (s' = 18) & (b1' = false) + <<A18_TT>> : (s' = 18) & (b1' = true);
  [] (s = 17) & !b0 -> <<A18_FF>> : (s' = 18) & (b1' = false) + <<A18_TF>> : (s' = 18) & (b1' = true);
  [] (s = 18) & b0 & b1 -> <<A19_FTT>> : (s' = 19) & (b2' = false) + <<A19_TTT>> : (s' = 19) & (b2' = true);
  [] (s = 18) & b0 & !b1 -> <<A19_FFT>> : (s' = 19) & (b2' = false) + <<A19_TFT>> : (s' = 19) & (b2' = true);
  [] (s = 18) & !b0 & b1 -> <<A19_FTF>> : (s' = 19) & (b2' = false) + <<A19_TTF>> : (s' = 19) & (b2' = true);
  [] (s = 18) & !b0 & !b1 -> <<A19_FFF>> : (s' = 19) & (b2' = false) + <<A19_TFF>> : (s' = 19) & (b2' = true);
  [] (s = 19) & !b0 & !b1 & !b2 -> true;
  [] (s = 19) & !b0 & !b1 & b2 -> true;
  [] (s = 19) & !b0 & b1 & !b2 -> true;
  [] (s = 19) & !b0 & b1 & b2 -> true;
  [] (s = 19) & b0 & !b1 & !b2 -> true;
  [] (s = 19) & b0 & !b1 & b2 -> true;
  [] (s = 19) & b0 & b1 & !b2 -> true;
  [] (s = 19) & b0 & b1 & b2 -> true;
endmodule
\end{lstlisting}

\subsection{Deutsch--Jozsa Balanced Oracle QMC}

\begin{lstlisting}[language=Haskell, breaklines=true, basicstyle=\tiny\ttfamily\linespread{0.5}, frame=single, framexleftmargin=1.5em]
qmc

const matrix A1_T = kron(M0, ID(8));
const matrix A1_F = kron(M1, ID(8));
const matrix A2 = kron(PauliX, ID(8));
const matrix A3_T = kron(kron(ID(2), M0), ID(4));
const matrix A3_F = kron(kron(ID(2), M1), ID(4));
const matrix A4 = kron(kron(ID(2), PauliX), ID(4));
const matrix A5_T = kron(kron(ID(4), M0), ID(2));
const matrix A5_F = kron(kron(ID(4), M1), ID(2));
const matrix A6 = kron(kron(ID(4), PauliX), ID(2));
const matrix A7_T = kron(ID(8), M0);
const matrix A7_F = kron(ID(8), M1);
const matrix A8 = kron(ID(8), PauliX);
const matrix A9 = kron(ID(8), PauliX);
const matrix A10 = kron(Hadamard, ID(8));
const matrix A11 = kron(kron(ID(2), Hadamard), ID(4));
const matrix A12 = kron(kron(ID(4), Hadamard), ID(2));
const matrix A13 = kron(ID(8), Hadamard);
const matrix A14 = [1, 0, 0, 0, 0, 0, 0, 0, 0, 0, 0, 0, 0, 0, 0, 0; 0, 0, 0, 0, 1, 0, 0, 0, 0, 0, 0, 0, 0, 0, 0, 0; 0, 1, 0, 0, 0, 0, 0, 0, 0, 0, 0, 0, 0, 0, 0, 0; 0, 0, 0, 0, 0, 1, 0, 0, 0, 0, 0, 0, 0, 0, 0, 0; 0, 0, 1, 0, 0, 0, 0, 0, 0, 0, 0, 0, 0, 0, 0, 0; 0, 0, 0, 0, 0, 0, 1, 0, 0, 0, 0, 0, 0, 0, 0, 0; 0, 0, 0, 1, 0, 0, 0, 0, 0, 0, 0, 0, 0, 0, 0, 0; 0, 0, 0, 0, 0, 0, 0, 1, 0, 0, 0, 0, 0, 0, 0, 0; 0, 0, 0, 0, 0, 0, 0, 0, 0, 1, 0, 0, 0, 0, 0, 0; 0, 0, 0, 0, 0, 0, 0, 0, 0, 0, 0, 0, 0, 1, 0, 0; 0, 0, 0, 0, 0, 0, 0, 0, 1, 0, 0, 0, 0, 0, 0, 0; 0, 0, 0, 0, 0, 0, 0, 0, 0, 0, 0, 0, 1, 0, 0, 0; 0, 0, 0, 0, 0, 0, 0, 0, 0, 0, 0, 1, 0, 0, 0, 0; 0, 0, 0, 0, 0, 0, 0, 0, 0, 0, 0, 0, 0, 0, 0, 1; 0, 0, 0, 0, 0, 0, 0, 0, 0, 0, 1, 0, 0, 0, 0, 0; 0, 0, 0, 0, 0, 0, 0, 0, 0, 0, 0, 0, 0, 0, 1, 0];
const matrix A15 = [1, 0, 0, 0, 0, 0, 0, 0, 0, 0, 0, 0, 0, 0, 0, 0; 0, 1, 0, 0, 0, 0, 0, 0, 0, 0, 0, 0, 0, 0, 0, 0; 0, 0, 1, 0, 0, 0, 0, 0, 0, 0, 0, 0, 0, 0, 0, 0; 0, 0, 0, 1, 0, 0, 0, 0, 0, 0, 0, 0, 0, 0, 0, 0; 0, 0, 0, 0, 0, 1, 0, 0, 0, 0, 0, 0, 0, 0, 0, 0; 0, 0, 0, 0, 1, 0, 0, 0, 0, 0, 0, 0, 0, 0, 0, 0; 0, 0, 0, 0, 0, 0, 0, 1, 0, 0, 0, 0, 0, 0, 0, 0; 0, 0, 0, 0, 0, 0, 1, 0, 0, 0, 0, 0, 0, 0, 0, 0; 0, 0, 0, 0, 0, 0, 0, 0, 1, 0, 0, 0, 0, 0, 0, 0; 0, 0, 0, 0, 0, 0, 0, 0, 0, 1, 0, 0, 0, 0, 0, 0; 0, 0, 0, 0, 0, 0, 0, 0, 0, 0, 1, 0, 0, 0, 0, 0; 0, 0, 0, 0, 0, 0, 0, 0, 0, 0, 0, 1, 0, 0, 0, 0; 0, 0, 0, 0, 0, 0, 0, 0, 0, 0, 0, 0, 0, 1, 0, 0; 0, 0, 0, 0, 0, 0, 0, 0, 0, 0, 0, 0, 1, 0, 0, 0; 0, 0, 0, 0, 0, 0, 0, 0, 0, 0, 0, 0, 0, 0, 0, 1; 0, 0, 0, 0, 0, 0, 0, 0, 0, 0, 0, 0, 0, 0, 1, 0];
const matrix A16 = kron(ID(4), [1, 0, 0, 0; 0, 1, 0, 0; 0, 0, 0, 1; 0, 0, 1, 0]);
const matrix A17 = kron(Hadamard, ID(8));
const matrix A18 = kron(kron(ID(2), Hadamard), ID(4));
const matrix A19 = kron(kron(ID(4), Hadamard), ID(2));
const matrix A20_F = kron(kron(ID(4), M0), ID(2));
const matrix A20_T = kron(kron(ID(4), M1), ID(2));
const matrix A21_FF = kron(kron(ID(2), M0), ID(4));
const matrix A21_TF = kron(kron(ID(2), M1), ID(4));
const matrix A22_FFF = kron(M0, ID(8));
const matrix A22_TFF = kron(M1, ID(8));
const matrix A22_FTF = kron(M0, ID(8));
const matrix A22_TTF = kron(M1, ID(8));
const matrix A21_FT = kron(kron(ID(2), M0), ID(4));
const matrix A21_TT = kron(kron(ID(2), M1), ID(4));
const matrix A22_FFT = kron(M0, ID(8));
const matrix A22_TFT = kron(M1, ID(8));
const matrix A22_FTT = kron(M0, ID(8));
const matrix A22_TTT = kron(M1, ID(8));

module dJozsaBal
  s: [0..22] init 0;
  b0: bool init false;
  b1: bool init false;
  b2: bool init false;

  [] (s = 0) -> <<A1_T>> : (s' = 1) & (b0' = true) + <<A1_F>> : (s' = 1) & (b0' = false);
  [] (s = 1) & b0 -> (s' = 2);
  [] (s = 1) & !b0 -> <<A2>> : (s' = 2);
  [] (s = 2) -> <<A3_T>> : (s' = 3) & (b0' = true) + <<A3_F>> : (s' = 3) & (b0' = false);
  [] (s = 3) & b0 -> (s' = 4);
  [] (s = 3) & !b0 -> <<A4>> : (s' = 4);
  [] (s = 4) -> <<A5_T>> : (s' = 5) & (b0' = true) + <<A5_F>> : (s' = 5) & (b0' = false);
  [] (s = 5) & b0 -> (s' = 6);
  [] (s = 5) & !b0 -> <<A6>> : (s' = 6);
  [] (s = 6) -> <<A7_T>> : (s' = 7) & (b0' = true) + <<A7_F>> : (s' = 7) & (b0' = false);
  [] (s = 7) & b0 -> (s' = 8);
  [] (s = 7) & !b0 -> <<A8>> : (s' = 8);
  [] (s = 8) -> <<A9>> : (s' = 9);
  [] (s = 9) -> <<A10>> : (s' = 10);
  [] (s = 10) -> <<A11>> : (s' = 11);
  [] (s = 11) -> <<A12>> : (s' = 12);
  [] (s = 12) -> <<A13>> : (s' = 13);
  [] (s = 13) -> <<A14>> : (s' = 14);
  [] (s = 14) -> <<A15>> : (s' = 15);
  [] (s = 15) -> <<A16>> : (s' = 16);
  [] (s = 16) -> <<A17>> : (s' = 17);
  [] (s = 17) -> <<A18>> : (s' = 18);
  [] (s = 18) -> <<A19>> : (s' = 19);
  [] (s = 19) -> <<A20_F>> : (s' = 20) & (b0' = false) + <<A20_T>> : (s' = 20) & (b0' = true);
  [] (s = 20) & b0 -> <<A21_FT>> : (s' = 21) & (b1' = false) + <<A21_TT>> : (s' = 21) & (b1' = true);
  [] (s = 20) & !b0 -> <<A21_FF>> : (s' = 21) & (b1' = false) + <<A21_TF>> : (s' = 21) & (b1' = true);
  [] (s = 21) & b0 & b1 -> <<A22_FTT>> : (s' = 22) & (b2' = false) + <<A22_TTT>> : (s' = 22) & (b2' = true);
  [] (s = 21) & b0 & !b1 -> <<A22_FFT>> : (s' = 22) & (b2' = false) + <<A22_TFT>> : (s' = 22) & (b2' = true);
  [] (s = 21) & !b0 & b1 -> <<A22_FTF>> : (s' = 22) & (b2' = false) + <<A22_TTF>> : (s' = 22) & (b2' = true);
  [] (s = 21) & !b0 & !b1 -> <<A22_FFF>> : (s' = 22) & (b2' = false) + <<A22_TFF>> : (s' = 22) & (b2' = true);
  [] (s = 22) & !b0 & !b1 & !b2 -> true;
  [] (s = 22) & !b0 & !b1 & b2 -> true;
  [] (s = 22) & !b0 & b1 & !b2 -> true;
  [] (s = 22) & !b0 & b1 & b2 -> true;
  [] (s = 22) & b0 & !b1 & !b2 -> true;
  [] (s = 22) & b0 & !b1 & b2 -> true;
  [] (s = 22) & b0 & b1 & !b2 -> true;
  [] (s = 22) & b0 & b1 & b2 -> true;
endmodule
\end{lstlisting}

\subsubsection{Deutsch--Jozsa Density Matrices}

\noindent{\scriptsize\texttt{r=[0;1;0;0;0;0;0;0;0;0;0;0;0;0;0;0;0;0]}}

{\bf Constant}

\noindent{\scriptsize\texttt{qeval(Q=? [F (s = 19 \& !b0 \& !b1 \& !b2)], r)}=}
\begin{lstlisting}[breaklines=true, basicstyle=\tiny\ttfamily\linespread{0.5}, firstline=43, lastline=59]
.5 .5 0 0 0 0 0 0 0 0 0 0 0 0 0 0
.5 .5 0 0 0 0 0 0 0 0 0 0 0 0 0 0
  0 0 0 0 0 0 0 0 0 0 0 0 0 0 0 0
  0 0 0 0 0 0 0 0 0 0 0 0 0 0 0 0
  0 0 0 0 0 0 0 0 0 0 0 0 0 0 0 0
  0 0 0 0 0 0 0 0 0 0 0 0 0 0 0 0
  0 0 0 0 0 0 0 0 0 0 0 0 0 0 0 0
  0 0 0 0 0 0 0 0 0 0 0 0 0 0 0 0
  0 0 0 0 0 0 0 0 0 0 0 0 0 0 0 0
  0 0 0 0 0 0 0 0 0 0 0 0 0 0 0 0
  0 0 0 0 0 0 0 0 0 0 0 0 0 0 0 0
  0 0 0 0 0 0 0 0 0 0 0 0 0 0 0 0
  0 0 0 0 0 0 0 0 0 0 0 0 0 0 0 0
  0 0 0 0 0 0 0 0 0 0 0 0 0 0 0 0
  0 0 0 0 0 0 0 0 0 0 0 0 0 0 0 0
  0 0 0 0 0 0 0 0 0 0 0 0 0 0 0 0
\end{lstlisting}

\noindent{\scriptsize\texttt{qeval(Q=? [F (s = 19 \& b0 \& b1 \& !b2)], r)}=}
\begin{lstlisting}[breaklines=true, basicstyle=\tiny\ttfamily\linespread{0.5}, firstline=61, lastline=77]
0 0 0 0 0 0 0 0 0 0 0 0 0 0 0 0
0 0 0 0 0 0 0 0 0 0 0 0 0 0 0 0
0 0 0 0 0 0 0 0 0 0 0 0 0 0 0 0
0 0 0 0 0 0 0 0 0 0 0 0 0 0 0 0
0 0 0 0 0 0 0 0 0 0 0 0 0 0 0 0
0 0 0 0 0 0 0 0 0 0 0 0 0 0 0 0
0 0 0 0 0 0 0 0 0 0 0 0 0 0 0 0
0 0 0 0 0 0 0 0 0 0 0 0 0 0 0 0
0 0 0 0 0 0 0 0 0 0 0 0 0 0 0 0
0 0 0 0 0 0 0 0 0 0 0 0 0 0 0 0
0 0 0 0 0 0 0 0 0 0 0 0 0 0 0 0
0 0 0 0 0 0 0 0 0 0 0 0 0 0 0 0
0 0 0 0 0 0 0 0 0 0 0 0 0 0 0 0
0 0 0 0 0 0 0 0 0 0 0 0 0 0 0 0
0 0 0 0 0 0 0 0 0 0 0 0 0 0 0 0
0 0 0 0 0 0 0 0 0 0 0 0 0 0 0 0
\end{lstlisting}

\noindent{\bf Balanced}

\noindent{\scriptsize\texttt{qeval(Q=? [F (s = 22 \& !b0 \& !b1 \& !b2)], r)}=}
\begin{lstlisting}[breaklines=true, basicstyle=\tiny\ttfamily\linespread{0.5}, firstline=5, lastline=21]
0 0 0 0 0 0 0 0 0 0 0 0 0 0 0 0
0 0 0 0 0 0 0 0 0 0 0 0 0 0 0 0
0 0 0 0 0 0 0 0 0 0 0 0 0 0 0 0
0 0 0 0 0 0 0 0 0 0 0 0 0 0 0 0
0 0 0 0 0 0 0 0 0 0 0 0 0 0 0 0
0 0 0 0 0 0 0 0 0 0 0 0 0 0 0 0
0 0 0 0 0 0 0 0 0 0 0 0 0 0 0 0
0 0 0 0 0 0 0 0 0 0 0 0 0 0 0 0
0 0 0 0 0 0 0 0 0 0 0 0 0 0 0 0
0 0 0 0 0 0 0 0 0 0 0 0 0 0 0 0
0 0 0 0 0 0 0 0 0 0 0 0 0 0 0 0
0 0 0 0 0 0 0 0 0 0 0 0 0 0 0 0
0 0 0 0 0 0 0 0 0 0 0 0 0 0 0 0
0 0 0 0 0 0 0 0 0 0 0 0 0 0 0 0
0 0 0 0 0 0 0 0 0 0 0 0 0 0 .5 .5
0 0 0 0 0 0 0 0 0 0 0 0 0 0 .5 .5
\end{lstlisting}

\noindent{\scriptsize\texttt{qeval(Q=? [F (s = 22 \& b0 \& b1 \& !b2)], r)}=}
\begin{lstlisting}[breaklines=true, basicstyle=\tiny\ttfamily\linespread{0.5}, firstline=23, lastline=39]
0 0 0 0 0 0 0 0 0 0 0 0 0 0 0 0
0 0 0 0 0 0 0 0 0 0 0 0 0 0 0 0
0 0 0 0 0 0 0 0 0 0 0 0 0 0 0 0
0 0 0 0 0 0 0 0 0 0 0 0 0 0 0 0
0 0 0 0 0 0 0 0 0 0 0 0 0 0 0 0
0 0 0 0 0 0 0 0 0 0 0 0 0 0 0 0
0 0 0 0 0 0 0 0 0 0 0 0 0 0 0 0
0 0 0 0 0 0 0 0 0 0 0 0 0 0 0 0
0 0 0 0 0 0 0 0 0 0 0 0 0 0 0 0
0 0 0 0 0 0 0 0 0 0 0 0 0 0 0 0
0 0 0 0 0 0 0 0 0 0 0 0 0 0 0 0
0 0 0 0 0 0 0 0 0 0 0 0 0 0 0 0
0 0 0 0 0 0 0 0 0 0 0 0 0 0 0 0
0 0 0 0 0 0 0 0 0 0 0 0 0 0 0 0
0 0 0 0 0 0 0 0 0 0 0 0 0 0 0 0
0 0 0 0 0 0 0 0 0 0 0 0 0 0 0 0
\end{lstlisting}

\subsection{Quantum Switch}

\begin{lstlisting}[language=Haskell, breaklines=true, basicstyle=\tiny\ttfamily\linespread{0.5}, frame=single, framexleftmargin=1.5em, firstline=81, lastline=154]
module switchCirc
  s: [0..39] init 0;
  b0: bool init false;
  b1: bool init false;
  b2: bool init false;
  b3: bool init false;

  [] (s = 0) -> <<A1>> : (s' = 1);
  [] (s = 1) -> <<A2>> : (s' = 2);
  [] (s = 2) -> <<A3>> : (s' = 3);
  [] (s = 3) -> <<A4>> : (s' = 4);
  [] (s = 4) -> <<A5>> : (s' = 5);
  [] (s = 5) -> <<A6>> : (s' = 6);
  [] (s = 6) -> <<A7>> : (s' = 7);
  [] (s = 7) -> <<A8>> : (s' = 8);
  [] (s = 8) -> <<A9>> : (s' = 9);
  [] (s = 9) -> <<A10>> : (s' = 10);
  [] (s = 10) -> <<A11>> : (s' = 11);
  [] (s = 11) -> <<A12>> : (s' = 12);
  [] (s = 12) -> <<A13>> : (s' = 13);
  [] (s = 13) -> <<A14>> : (s' = 14);
  [] (s = 14) -> <<A15>> : (s' = 15);
  [] (s = 15) -> <<A16>> : (s' = 16);
  [] (s = 16) -> <<A17>> : (s' = 17);
  [] (s = 17) -> <<A18>> : (s' = 18);
  [] (s = 18) -> <<A19>> : (s' = 19);
  [] (s = 19) -> <<A20>> : (s' = 20);
  [] (s = 20) -> <<A21>> : (s' = 21);
  [] (s = 21) -> <<A22>> : (s' = 22);
  [] (s = 22) -> <<A23>> : (s' = 23);
  [] (s = 23) -> <<A24>> : (s' = 24);
  [] (s = 24) -> <<A25>> : (s' = 25);
  [] (s = 25) -> <<A26>> : (s' = 26);
  [] (s = 26) -> <<A27>> : (s' = 27);
  [] (s = 27) -> <<A28>> : (s' = 28);
  [] (s = 28) -> <<A29>> : (s' = 29);
  [] (s = 29) -> <<A30>> : (s' = 30);
  [] (s = 30) -> <<A31>> : (s' = 31);
  [] (s = 31) -> <<A32>> : (s' = 32);
  [] (s = 32) -> <<A33>> : (s' = 33);
  [] (s = 33) -> <<A34>> : (s' = 34);
  [] (s = 34) -> <<A35>> : (s' = 35);
  [] (s = 35) -> <<A36_F>> : (s' = 36) & (b0' = false) + <<A36_T>> : (s' = 36) & (b0' = true);
  [] (s = 36) & b0 -> <<A37_FT>> : (s' = 37) & (b1' = false) + <<A37_TT>> : (s' = 37) & (b1' = true);
  [] (s = 36) & !b0 -> <<A37_FF>> : (s' = 37) & (b1' = false) + <<A37_TF>> : (s' = 37) & (b1' = true);
  [] (s = 37) & b0 & b1 -> <<A38_FTT>> : (s' = 38) & (b2' = false) + <<A38_TTT>> : (s' = 38) & (b2' = true);
  [] (s = 37) & b0 & !b1 -> <<A38_FFT>> : (s' = 38) & (b2' = false) + <<A38_TFT>> : (s' = 38) & (b2' = true);
  [] (s = 37) & !b0 & b1 -> <<A38_FTF>> : (s' = 38) & (b2' = false) + <<A38_TTF>> : (s' = 38) & (b2' = true);
  [] (s = 37) & !b0 & !b1 -> <<A38_FFF>> : (s' = 38) & (b2' = false) + <<A38_TFF>> : (s' = 38) & (b2' = true);
  [] (s = 38) & b0 & b1 & b2 -> <<A39_FTTT>> : (s' = 39) & (b3' = false) + <<A39_TTTT>> : (s' = 39) & (b3' = true);
  [] (s = 38) & b0 & b1 & !b2 -> <<A39_FFTT>> : (s' = 39) & (b3' = false) + <<A39_TFTT>> : (s' = 39) & (b3' = true);
  [] (s = 38) & b0 & !b1 & b2 -> <<A39_FTFT>> : (s' = 39) & (b3' = false) + <<A39_TTFT>> : (s' = 39) & (b3' = true);
  [] (s = 38) & b0 & !b1 & !b2 -> <<A39_FFFT>> : (s' = 39) & (b3' = false) + <<A39_TFFT>> : (s' = 39) & (b3' = true);
  [] (s = 38) & !b0 & b1 & b2 -> <<A39_FTTF>> : (s' = 39) & (b3' = false) + <<A39_TTTF>> : (s' = 39) & (b3' = true);
  [] (s = 38) & !b0 & b1 & !b2 -> <<A39_FFTF>> : (s' = 39) & (b3' = false) + <<A39_TFTF>> : (s' = 39) & (b3' = true);
  [] (s = 38) & !b0 & !b1 & b2 -> <<A39_FTFF>> : (s' = 39) & (b3' = false) + <<A39_TTFF>> : (s' = 39) & (b3' = true);
  [] (s = 38) & !b0 & !b1 & !b2 -> <<A39_FFFF>> : (s' = 39) & (b3' = false) + <<A39_TFFF>> : (s' = 39) & (b3' = true);
  [] (s = 39) & !b0 & !b1 & !b2 & !b3 -> true;
  [] (s = 39) & !b0 & !b1 & !b2 & b3 -> true;
  [] (s = 39) & !b0 & !b1 & b2 & !b3 -> true;
  [] (s = 39) & !b0 & !b1 & b2 & b3 -> true;
  [] (s = 39) & !b0 & b1 & !b2 & !b3 -> true;
  [] (s = 39) & !b0 & b1 & !b2 & b3 -> true;
  [] (s = 39) & !b0 & b1 & b2 & !b3 -> true;
  [] (s = 39) & !b0 & b1 & b2 & b3 -> true;
  [] (s = 39) & b0 & !b1 & !b2 & !b3 -> true;
  [] (s = 39) & b0 & !b1 & !b2 & b3 -> true;
  [] (s = 39) & b0 & !b1 & b2 & !b3 -> true;
  [] (s = 39) & b0 & !b1 & b2 & b3 -> true;
  [] (s = 39) & b0 & b1 & !b2 & !b3 -> true;
  [] (s = 39) & b0 & b1 & !b2 & b3 -> true;
  [] (s = 39) & b0 & b1 & b2 & !b3 -> true;
  [] (s = 39) & b0 & b1 & b2 & b3 -> true;
endmodule
\end{lstlisting}

\end{appendices}
\end{document}